\documentclass[12pt]{article}

\usepackage{xcolor}

\usepackage{amsthm}
\usepackage{dsfont}
\usepackage{eurosym}
\usepackage{amsmath}
\usepackage{amsfonts}
\usepackage{makeidx}
\usepackage{amssymb}
\usepackage{times}
\usepackage{anysize}
\usepackage{pst-node}
\usepackage{pst-tree}
\usepackage{lscape}
\usepackage{graphics,graphicx}
\usepackage{hyperref}
\usepackage{cite}
\usepackage{enumerate}
\usepackage{enumitem}
\usepackage{float}
\usepackage{tikz}
\usetikzlibrary{matrix,decorations.pathreplacing,calc}

\pgfkeys{tikz/mymatrixenv/.style={decoration=brace,every left delimiter/.style={xshift=3pt},every right delimiter/.style={xshift=-3pt}}}
\pgfkeys{tikz/mymatrix/.style={matrix of math nodes,left delimiter=[,right delimiter={]},inner sep=2pt,column sep=1em,row sep=0.5em,nodes={inner sep=0pt}}}
\pgfkeys{tikz/mymatrixbrace/.style={decorate,thick}}

\newcommand{\rvline}{\hspace*{-\arraycolsep}\vline\hspace*{-\arraycolsep}}

\def\pstree@balancedfit#1#2{\edef\next{\noexpand\pstree@@balancedfit#1\noexpand\@nil#2\noexpand\@nil}\next
\ifnum\pst@cntg=\z@
\pstree@max{#1}\pst@cnth
\else
\pstree@max{#2}\pst@cnth
\fi
\advance\pst@cnth\pst@cnth
\advance\pst@cnth\psk@thistreesep\relax
\advance\pst@cnth\pstree@tspace\relax
\gdef\pstree@tspace{\z@}}
\def\pstree@@balancedfit#1,#2\@nil#3,#4\@nil{\ifnum#1=\pstree@stop
\let\next\relax
\pst@cntg=\@ne
\else
\ifnum#3=\pstree@stop
\let\next\relax
\pst@cntg=\z@
\else
\def\next{\pstree@@balancedfit#2\@nil#4\@nil}\fi
\fi
\next}

\psset{fillstyle=solid}

\def\h{\mathfrak h}

\newtheorem{theorem}{Theorem}[section]
\newtheorem{corollary}[theorem]{Corollary}\newtheorem{lemma}[theorem]{Lemma}

\marginsize{2cm}{2cm}{1cm}{1cm}
\newtheorem{prop}[theorem]{Proposition}

\newcommand{\cod}{\operatorname{cod}}

\newcommand{\tgsp}{T^{\operatorname{ext}}}
\theoremstyle{definition}
\newtheorem{defn}[theorem]{Definition}
\newtheorem*{remark}{Remark}

\input{tcilatex}

\makeatother

\begin{document}
\title{Catastrophe Theory for $\Gamma$-invariant Unfoldings with Applications
to Quantum Many-body Theory}
\author{K. Rodrigues Alves}
\maketitle

\begin{abstract}
\bigskip{}

The theory of singularities and its broad ramifications, especially catastrophe theory, have found fertile ground in some areas of physics (e.g., caustics, wave optics) for their applications. In the context of quantum many-body theory, however, their results, despite being useful, are not generally known by the scientific community and often require a non-trivial adaptation, mainly due to the effects of symmetries in the associated physical system, which are not encompassed by the original theory. In this article, we provide an extension of the main results of Ren\'{e} Thom's catastrophe theory for the case of germs and unfoldings possessing special symmetries. In a more mathematically precise language, we provide a proof of the determinacy theorems for germs that are invariant under the action of an arbitrary compact Lie group, and of the transversality and stability theorems for the case of invariant unfoldings. The results obtained can be seen as an extension and adaptation of the works of {[}7{]} and {[}8{]} on singularity theory of $\mathbb{R}^{n}$ to $\mathbb{R}^{n}$ mappings to the particular scenario of catastrophe theory. Finally, we also provide a classification theorem for unfoldings invariant under the action of $\mathbb{Z}_{2}$, and present some of the possible applications of the theory to the study of phase transitions in quantum many-body systems.
\end{abstract}
\tableofcontents{}\newpage{}

\section{Introduction}

During the years 1950s, the theory of classification of singularities
of smooth mappings has seen an outstanding progress. Such progress was mainly
a consequence of the seminal work developed by H. Whitney, who in 1952 published an article
on the classification of mappings from $\mathbb{R}^{2}$ to $\mathbb{R}^{2}$
\cite{whitney1955singularities}, that laid the basis for many subsequent
works on the study of singularities of mappings. One of these works,
in particular, carried out by the mathematician Ren\'e Thom \cite{thom1974stabilite},
now receives the name of catastrophe theory. To first understand what
kind of classification of singularities Thom's work is concerned with,
we may begin by looking at mappings of functions $f$ from $U\subset\mathbb{R}^{n}$
to $\mathbb{R}$. For a \textit{regular point} $x_{0}\in U$ of $f$,
i.e., a point $x_{0}$ such that the gradient $\Vec{\nabla}f(x_{0})$
is non-degenerate, a consequence of the implicit function theorem
is that there exists a diffeomorphism $\phi$ such that, in its coordinates,
$f$ is given by 
\begin{equation}
(f\circ\phi)(x)=x_{1}.
\end{equation}
If, instead, $x_{0}$ is not a regular point but rather a \textit{non-degenerate
critical point}, i.e., a point where the gradient $\Vec{\nabla}f(x_{0})$
is degenerate but the Hessian $\operatorname{Hess}(f)(x_{0})$ is not, then similarly
to the previous case, by the so-called Morse Lemma it follows that
there exists again a diffeomorphism $\phi$ such that, around $x_{0}$, $f$ is
given by 
\begin{equation}
(f\circ\phi)(x)=\pm x_{1}^{2}+\cdots\pm x_{n}^{2}.
\end{equation}
The qualitative analysis of the behaviour of $f$ around $x_{0}$
can be, therefore, easily done in these two cases, since the diffeomorphism
$\phi$ gives us a great deal of simplification. Any function
near a regular or a non-degenerate critical point follows into one
of the examples above. However, the situation 
changes if $x_{0}$ is now a \textit{degenerate critical point}; a
critical point with a degenerate Hessian. Then, the behavior
of $f$ around $x_{0}$ ceases to be trivial, and we need to enter
more deeply in the theory of singularities in order to understand it. 


The brief analysis described above can be restated in a more general and abstract way, as in the following consideration: we start with the space of
smooth real-valued functions on $\mathbb{R}^{n}$, which we simply
denote by $S$, and an equivalence relation associated with the (infinite-dimensional) Lie group $G$ of smooth diffeomorphisms on $\mathbb{R}^{n}$
preserving the origin, defined by $f\sim g$ iff there exists
a $\phi\in G$ such that $f=g\circ\phi$. The equivalence relation
is supposed to preserve the properties that we are interested in;
therefore, given a function $f\in S$, one of the goals is to find
a function $g$, belonging to the same equivalence class of $f$, such that $g$ is an element we
can explicitly compute the relevant properties of our system, which
in turn will still hold true for $f$. Another related important concept
is that of \textit{stability}: $f$ is stable when there exists a
neighborhood of $f$ (in a suitable topology) where every point is equivalent to $f$ itself.
A remarkable discovery of singularity theory is that we can learn
about the stability of $f$ by analyzing the differential of the action
of $G$ into $S$, i.e., the differential $Dr_{f}$ of the mapping
$r_{f}:G\rightarrow S$ given by $r_{f}(\phi)=f\circ\phi$. It follows
that $f$ is stable if and only if $Dr_{f}$ is surjective. And this
happens if and only if $\vec{\nabla}f$ or $\operatorname{Hess}(f)$ are non-degenerate;
in these scenarios, $f$ falls into one of the two cases shown above.


Suppose now that $Dr_{f}$ is not surjective, i.e., that $f$ is not
a stable point. Then, we can consider a submanifold $F$ passing through
$f$, such that its tangent space $T_F$ at $f$ is a complement of the image of the differential $Dr_{f}$. Note that $F$ can also be thought of as a parametric family of functions in $S$. It follows that although $f$ is not stable, $F$, by intersecting
transversally the manifold generated by $r_{f}$, is itself stable -- for an
associated equivalence relation. In fact, this relationship between transversality and stability is one of the main theorems
about catastrophe theory. The number of extra dimensions that the
differential $Dr_{f}$ fails to generate is called the \textit{codimension}
of the function. This is what is in the core of catastrophe theory:
to show that such transversal families of functions are stable, and
equivalent to each other. Moreover, for functions with codimension
smaller than or equal to $6$, it can also be shown that there exists
only a finite number of equivalence classes, which in turn implies
the existence of only a finite number of equivalence classes of parametric
families of functions; a result that is famously known as \textit{Thom's Classification
Theorem}. The polynomial representatives of these equivalence classes
are usually called elementary catastrophes in the literature. \footnote{The first list of the elementary catastrophes went only up to codimension
$4$. Later, however, extensions to the initial classification were
provided, and from codimension $7$ onwards it was shown that the
total number of functions in the classification theorem already becomes
infinite. To the author's knowledge, the largest classification table
is exposed in \cite{siersma1974classification}.} There are, however, certain limitations of the original theory regarding
families of functions $F$ possessing special symmetries. The symmetry
of the functions in $F$ can prevent its tangent space from generating
the missing complementary vectors, and hence, it would not be stable.
However, if we in turn restrict $G$ to only diffeomorphisms that
preserve the desired symmetry (also called \textit{equivariant} diffeomorphisms), and rescrict the
function space $S$ to just functions that possess the assumed symmetry, thus
limiting the perturbations to only symmetric perturbations, one may
then recover similar results for the symmetric family $F$. Such extensions
of the original theory of singularities has been studied in a similar
albeit different context by \cite{meyer1986singularities} for some
kinds of symmetry, and finally by \cite{damon1984unfolding} in more
generality. Their studies, however, did not focus on catastrophe theory,
but rather to general mappings $F:\mathbb{R}^{n}\rightarrow\mathbb{R}^{k}$, with a different equivalence relation
and also a different symmetry: rather than being \textit{invariant} under a
symmetry group, they are what is called \textit{equivariant}. Here, we intend
to present a concise and complete exposure of these results adapted
to the context of catastrophe theory.

Thom's original goal for developing catastrophe theory, introduced
to the scientific community in his book \cite{thom1974stabilite},
was to provide a general theory for classifying the behaviour of discontinuous
phenomena coming from a continuous parameter change. In his setup,
given a system with $r$ control parameters, the "state" of the
system (which will depend on the values of the parameters) is described
by the minima $x_{0}\in\mathbb{R}^{n}$ of a function $F:(x,u)\mapsto F(x,u)\in\mathbb{R}$,
i.e., the set of possible states of the system for some fixed value of $u\in \mathbb{R}^r$ are described by the
points of the set 
\begin{equation}
S_{u}=\{x_{0}\in\mathbb{R}^{n}\;|\;F(x_{0},u)=\min_{x}F(x,u)\},\;\;u\in\mathbb{R}^{r}.
\end{equation}
As the parameter $u$ moves, so do the points in $S_{u}$, possibly
leading to an eventual sudden jump in the states of the system, i.e.,
a catastrophe. Morover, as mentioned before, depending on the tangent space generated by $F$ at $u=0$, it is also possible to show that there
is only a finite number of equivalence classes of functions that $F$ must belong to; a result known as Thom's Classification Theorem. Thom's works, together with Zeeman's, had a great impact at
the time in areas outside of exact sciences, because of their claims
of being able to correctly model, with only a finite number of functions, an incredible amount of different
phenomena: from prison uprisings to a dog's aggressiveness. After
this initial excitement, however, a stronger criticism came \cite{guckenheimer1978catastrophe},
questioning the validity of modelling real-world phenomena in such
a way, and at last the theory fell into disuse outside of more mathematical
areas.

We, in particular, avoid this criticism here by considering the application of catastrophe
theory not as a generic theory for modelling discontinuous processes,
but as a direct application of singularity theory to the study of
phase transitions in classical and quantum mechanics, where the function
$F$ does not appear as an a priori consideration of the physical
system, but as a direct consequence of the proposed model. The results
coming from catastrophe theory have relevant contributions, especially
in proving the existence of phase transitions and the stability of
phase diagrams. As an example, we may consider the BCS model of superconductivity.
In the so-called strong-coupling limit, and disregarding the kinetic
energy of the electrons in the lattice, it can be shown \cite{bru2010effect}
that the thermodynamical pressure of the system converges, in the
infinite-volume limit, to 
\begin{equation}
P(u_{1},u_{2})=-\inf_{c\in\mathbb{R}}F(c,u_{1},u_{2}),\;\;F(c,u_{1},u_{2})=u_{1}c^{2}-\ln\left(1+e^{-u_{2}}\cosh\left(\sqrt{u_{2}^{2}+u_{1}^{2}c^{2}}\right)\right),
\end{equation}
where the parameters $u_{1},u_{2}$ are associated with the strengths
of the electron-electron interactions in our system, and the variable
$c$ is associated with different phases of our system. By a suitable
change of coordinates (that in particular must not mix the parameters
and the variable), one may show that the function $F$ is \textit{equivalent}
(in a suitable sense, via diffeomorphisms that preserve its minima
and critical points, as we shall expose) to 
\begin{equation}
\tilde{F}(c,u_{1},u_{2})=c^{6}+u_{1}c^{2}+u_{2}c^{4},
\end{equation}
and hence, its set of bifurcations of states can be easily studied
by the analysis of the above polynomial function, giving rise to a
phase diagram of the system.

It follows, however, that the above example does not fall into the
original classification proposed by Thom; and the reason behind it
is the existence of symmetries in the family $F$. The symmetry that
plays a role here is the parity symmetry of $F(\cdot,u_{1},u_{2})$,
which is associated with the $U(1)$ Gauge symmetry of the respective
Hamiltonian of the quantum system. Families of functions with some
special symmetry in its variables are in general not stable, since
even an arbitrarily small non-symmetric perturbation can break the
symmetry of the function and, as a consequence, generate different
singularities and catastrophes that were not present in the original
symmetric case. However, they can be considered stable, in a slightly
different sense, if we restrict the perturbations only to those who
also possess the same symmetry of the functions in consideration.
Then, we may recover the standard results of catastrophe theory, but
for functions possessing a special symmetry, which are abundant in
the context of physics. This is thus the aim of this paper, to expose
to a broader audience the features of catastrophe theory that can
be explored in a mathematically precise way, in contexts encountered
in quantum statistical mechanics, where the family $F$ has a special
symmetry, and hence the original theory cannot be applied. To the author's
knowledge, there is no other work in quantum statistical mechanics
that uses this area of singularity theory as a powerful mathematical
tool to extract relevant qualitative behaviours of the system. In
a further work, we aim at using the results exposed here to prove
the existence and stability of phase transitions to a large class
of mean-field models for fermionic systems.

The paper is divided as follows: In the second chapter, we prove the
so-called finite determinacy theorem (Theorem 2.3) for germs with
special symmetries. 
In the third chapter we show the transversality and stability theorems
(Theorems 3.6 and 3.7, respectively) for families of functions possessing symmetries, and finally in chapter $4$ we present
a classification theorem (Theorem 4.4) for the special case of $\mathbb{Z}_{2}$-symmetric
families of functions.

\section{$\Gamma$-invariant germs}

\subsection{Basic definitions}

Here we start with the basic definitions necessary for the development of the
results. We note first that catastrophe theory deals with singularities
\textit{locally} i.e., in a neighborhood of a point. Hence, the natural
language for the theory is the language of germs. We introduce now
certain germ spaces that will be used throughout the text. We keep
a similar notation to the one of \cite{damon1984unfolding}.

\begin{defn}\label{defbasic} Let $U\subset\mathbb{R}^{n}$ be an open neighborhood
of $0$. We define:
\begin{enumerate}[label=\roman*.)]
\item $\mathcal{E}_{n}$ as the set of all germs at $0$ of functions in
$C^{\infty}(U,\mathbb{R})$, and $\vec{\mathcal{E}}_{n,m}$ as the set
of all germs at $0$ in $C^{\infty}(U,\mathbb{R}^{m})$,
\item $\mathcal{M}^k_{n}$ as the set of all germs at $0$ of functions in
$C^{\infty}(U,\mathbb{R})$ whose derivatives up to order $k-1$
vanish at $0$, and $\vec{\mathcal{M}}_{n,m}^{k}$ as the set
of all germs at $0$ in $C^{\infty}(U,\mathbb{R}^{m})$ whose derivatives up to order $k-1$
vanish at $0$,
\item $J_{n}^{k}$ as the set of all germs of polynomials in $C^{\infty}(U,\mathbb{R})$
of order up to $k$, and $J_{0,n}^{k}\subset J_{n}^{k}$ as the set
of germs of polynomials in $J_{n}^{k}$ with no constant term. 
\end{enumerate}
\end{defn}

\begin{remark}
We recall that, for an open $U\subset\mathbb{R}^{n}$ and a function $f\in C^{\infty}(U,\mathbb{R})$ the germ of $f$ at $p\in U$ is the equivalence class of all functions $g\in C^{\infty}(U,\mathbb{R})$ such that $g=f$ for some neighborhood $V\subset U$ of $p$.
\end{remark}

In order to include the symmetries in the germs, we, for now on, let $\Gamma$ be a compact Lie group, acting orthogonally on $\mathbb{R}^{n}$. Since we will only be interested in the action of $\Gamma$, to simplify the notation, the matrix representative in $GL(n)$ of any element $\gamma\in \Gamma$ will also be denoted by the same letter $\gamma$.

\begin{defn}\label{defsym} Let $\Gamma$ be a compact Lie group acting orthogonally on $\mathbb{R}^{n}$. A germ $f\in\mathcal{E}_{n}$ satisfying 
\[
f(\gamma x)=f(x)\quad\text{ for all }\gamma\in\Gamma
\]
will be called \textit{$\Gamma$-invariant}. A germ $\phi\in\vec{\mathcal{E}}_{n}$
satisfying

\[
\phi(\gamma x)=\gamma\phi(x)\quad\text{ for all }\gamma\in\Gamma
\]
will be called \textit{$\Gamma$-equivariant}. The set of all $\Gamma$-invariant
and $\Gamma$-equivariant germs will be denoted, respectively, by
$\mathcal{E}_{n}(\Gamma)$ and $\vec{\mathcal{E}}_{n}(\Gamma)$. In
particular, we denote the set of $\Gamma$-invariant polynomials by $P_n(\Gamma)\subset \mathcal{E}_n(\Gamma)$ and the set of $\Gamma$-equivariant germ diffeomorphisms preserving the origin
by $L_{n}(\Gamma)\subset \vec{\mathcal{E}}_n(\Gamma)$. \end{defn}

\begin{remark} One of the reasons to restrict the theory to compact Lie
groups is to be able to apply in that case the so-called Hilbert basis theorem, which states that the sub-algebra
$P_n(\Gamma)$ of $\Gamma$-invariant polynomial germs is generated (as an algebra)
by finitely many $\Gamma$-invariant polynomials. A similar result also states that $\vec{\mathcal{E}}_n(\Gamma)$ is generated, as an $\mathcal{E}_n(\Gamma)$-module, by finitely many homogeneous
$\Gamma$-equivariant polynomials.
\end{remark}

The following definition gives us the relevant equivalence class of germs in $\mathcal{E}_n(\Gamma)$, which will in turn give rise to the classification theorem.

\begin{defn}\label{defeq1}
Given $f,g\in\mathcal{E}_{n}(\Gamma)$, if there exists
some $\phi\in L_{n}(\Gamma)$ such that $g=f\circ\phi$, we say that
$f$ and $g$ are $\Gamma$-equivalent.
\end{defn}

\subsection{The tangent space of $\Gamma$-invariant germs}

We now analyze the tangent space of the orbit of a germ $f$ with
respect to the action of $L_n(\Gamma)$ (recall Definition \ref{defsym}). As stated in the introduction, the tangent space generated by the orbit of $f$ with respect to $L_n(\Gamma)$ is associated with the stability of our germ $f$, as we shall see later. Hence, let $f\in\mathcal{E}_{n}(\Gamma)$.
The orbit of $f$ with respect to the action of the (infinite-dimensional
Lie) group $L_{n}(\Gamma)$ is denoted by 
\[
\mathcal{O}_{f}\doteq\{f\circ\phi\;|\;\phi\in L_{n}(\Gamma)\}.
\]
Let $\mathbf{1}$ be the identity of $L_{n}(\Gamma)$. A tangent vector
at $\mathbf{1}$ is then of the type $(\mathbf{1}+t\phi)(x)$, where
$\phi\in\vec{\mathcal{M}}_{n}(\Gamma)$, since one must have $\phi(0)=0$.
Hence, an element of the tangent space $T_{\mathcal{O}_f}$ of $\mathcal{O}_{f}$
is given by

\[
\frac{df(\mathbf{1}+t\phi)(x)}{dt}\bigg\vert_{t=0}=\sum_{i=1}^{n}\phi_{i}(x)\frac{\partial f(x)}{\partial x_{i}}.
\]
The tangent space $T_{\mathcal{O}_f}$ is consequently given by

\[
T_{\mathcal{O}_f}=\{\phi\cdot\vec{\nabla}f\;|\;\phi=(\phi_{1},\dots,\phi_{n})\in\vec{\mathcal{M}}_{n}(\Gamma)\}.
\]
We formalize it in the next definition:

\begin{defn}\label{deftg}
Let $f\in \mathcal{E}_n(\Gamma)$. Its tangent space with respect to the action of $L_n(\Gamma)$ is given by
\begin{equation}
T_{\mathcal{O}_f}=\{\phi\cdot\vec{\nabla}f\;|\;\phi=(\phi_{1},\dots,\phi_{n})\in\vec{\mathcal{M}}_{n}(\Gamma)\}.
\end{equation}
We also define the \textit{extended} tangent space by
\begin{equation}
\tgsp_{\mathcal{O}_f}=\{\phi\cdot\vec{\nabla}f\;|\;\phi=(\phi_{1},\dots,\phi_{n})\in\vec{\mathcal{E}}_{n}(\Gamma)\}.
\end{equation}
\end{defn}

\begin{prop} For any $f\in\mathcal{E}_{n}(\Gamma)$, $T_{\mathcal{O}_f}\subset \mathcal{M}_{n}(\Gamma)$
and $\tgsp_{\mathcal{O}_f}\subset \mathcal{E}_{n}(\Gamma)$.
\end{prop}

\begin{proof} The inclusions  $T_{\mathcal{O}_f}\subset \mathcal{M}_{n}$ and $\tgsp_{\mathcal{O}_f}\subset \mathcal{E}_{n}$ are obvious. Furthermore, by the chain rule, note that
\begin{equation}
\frac{\partial(f(\gamma x))}{\partial x_{i}}=\sum_{j=1}^{n}\frac{\partial f(\gamma x)}{\partial x_{j}}\gamma_{j,i},
\end{equation}
where $\gamma_{i,j}$ are the matrix elements of $\gamma$ at row $i$ and column $j$. I.e., one has
\begin{equation}
\vec{\nabla}(f(\gamma x))=\gamma^{T}\vec{\nabla}f(\gamma x).
\end{equation}
Inverting the equation, making use of the fact that $\Gamma$ acts
orthogonally on $\mathbb{R}^{n}$, and that $f(\gamma x)=f(x)$, it
follows that
\begin{equation}
\vec{\nabla}f(\gamma x)=\gamma\vec{\nabla}f(x).
\end{equation}
Hence, if $\phi\in\vec{\mathcal{E}}_{x}(\Gamma)$, one has
\begin{equation}
\phi(\gamma x)\cdot\vec{\nabla}f(\gamma x)=((\gamma^{T}\gamma\phi(x))\cdot\vec{\nabla}f(x)=\phi(x)\cdot\vec{\nabla}f(x),
\end{equation}
i.e., $\phi\cdot\vec{\nabla}f$ is $\Gamma$-invariant. Hence, the propoition follows.
\end{proof}

\begin{remark} For a given family $F(x,u)$ with $f(x)=F(x,0)$,
the number of parameters $u$ that $F$ needs for it to be stable
is associated with the size of the tangent space $\tgsp_{\mathcal{O}_f}$, this is why $\tgsp_{\mathcal{O}_f}$ is
so important. We note here that the above tangent space is infinite-dimensional, and eventually
we will need to reduce it to a finite-dimensional manifold, by considering
the orbit in the jet space $J_{n}^{k}(\Gamma)$, where we then can
make use of standard results in the theory of euclidean manifolds.
\end{remark}

\subsection{Determinacy of $\Gamma$-invariant germs}

To apply theorems about transversality, coming from standard differential
geometry, it will be required for us to work on a finite-dimensional
manifold. For this, we prove the so-called determinacy theorem, which allows us to work on the jet space of functions up to some finite order
$k$.

\begin{defn} Let $f\in\mathcal{E}_{n}(\Gamma)$. Then, $f$ is said
to be $\Gamma$-$k$-determined if $f$ is $\Gamma$-equivalent to any
other $g\in\mathcal{E}_n(\Gamma)$ such that $j^{k}f=j^{k}g$, where $j^kf$ is the germ of the Taylor expansion of $f$ up to order $k$. Moreover, $f$ is said to be $\Gamma$-finitely-determined if it is
$\Gamma$-$k$-determined for some $k<\infty$.
\end{defn}

We note that if $f\in \mathcal{E}_n(\Gamma)$, then $j^kf\in \mathcal{E}_n(\Gamma)$ for any $k\in\mathbb{N}$. To see this, first recall that $j^kf$ is the unique polynomial that satisfies
\begin{equation*}
\lim_{\Vert x\Vert\rightarrow0}\frac{\vert f(x)-j^kf(x)\vert}{\Vert x\Vert^k}=0.
\end{equation*}
Since $\Gamma$ acts orthogonally on $\mathbb{R}^n$ and $f$ is $\Gamma$-invariant, one also has
\begin{equation*}
\lim_{\Vert x\Vert\rightarrow0}\frac{\vert f(x)-j^kf(\gamma x)\vert}{\Vert x\Vert^k}=0,\quad \text{for all }\gamma\in\Gamma.
\end{equation*}
In particular, by the uniqueness of $j^kf$, it follows that
\begin{equation*}
j^kf(x)=j^kf(\gamma x),\quad \text{for all }\gamma\in \Gamma,
\end{equation*}
i.e., $j^kf$ is $\Gamma$-invariant.

\begin{defn}\label{defcod} Let $f\in\mathcal{E}_{n}(\Gamma)$. The $\Gamma$-codimension
of $f$ is defined as
\begin{equation}
\cod_{\Gamma}(f)\doteq\dim_{\mathbb{R}}\frac{\mathcal{M}_{n}(\Gamma)}{\tgsp_{\mathcal{O}_f}}.
\end{equation}
\end{defn}

\begin{remark} Recall the definition of $\tgsp_{\mathcal{O}_f}$ (Equation (\ref{deftg})). The codimension
of a germ is related to the degree of degeneracy of its critical point,
and consequently to the number of parameters necessary to have an
unfolding of the germ stable under perturbations. In particular, only germs with finite $\Gamma$-codimension admit stable unfoldings, as we shall see later. To prove the main
theorem about the determinacy of $\Gamma$-invariant germs, which
states that finite codimension implies finite determinacy and vice-versa,
we need a generalization for the scenario of invariant germs of whats
is known as Mather's Lemma in the non-invariant case \cite{mather1968stability},
and such generalization is given in terms of \textit{DA-algebras}.
\end{remark}

\begin{defn} A \textit{differentiable algebra (DA-algebra)} is an
$\mathbb{R}$-algebra $A$, together with a surjective homomorphism
$\phi:\mathcal{E}_m\rightarrow A$, for some $m\in\mathbb{N}$. \end{defn}

\begin{theorem}[Schwarz theorem] Let $\Gamma$ be a compact Lie group acting
orthogonally on $\mathbb{R}^{n}$, and let $\rho_{1},\dots,\rho_{m}$
be the generators of $P_{n}(\Gamma)$. Then, it follows that, for
any germ $f\in\mathcal{E}_n(\Gamma)$, there exists some $h\in\mathcal{E}_m(\Gamma)$
such that
\begin{equation}
f(x)=h(\rho(x)),\;\text{ where }\;\rho(x)=(\rho_{1}(x),\dots,\rho_{k}(x)).
\end{equation}
\end{theorem}

\begin{remark} We recall that for any compact Lie group $\Gamma$, the set $P_{n}(\Gamma)$ of $\Gamma$-invariant polynomials has finite
polynomial generators $\rho_1,\dots,\rho_m$, by Hilbert Basis theorem (see \cite{poenaru2006singularites}, Theorem 1). Therefore, by choosing the homomorphism $\rho^*:\mathcal{E}_m\rightarrow \mathcal{E}_n(\Gamma)$ defined as
\begin{equation}
\rho^*(h)=h\circ \rho,\quad\rho(x)=(\rho_{1}(x),\dots,\rho_{k}(x)),
\end{equation}
it clearly follows, from Schwarz theorem stated above, that the pair $(\mathcal{E}_{x}(\Gamma),\rho^*)$ is a DA-algebra.
\end{remark}

We now state the generalization of Mather's lemma, whose proof we omit here due to its technicality, but it can be found in \cite[Lemma 7.3]{damon1984unfolding}.

\begin{theorem}[Generalized Mather's Lemma]\label{mlemma} Let $\mathfrak{R}=(R,\phi)$
be a DA-algebra, $a:N\rightarrow M$ be a homomorphism of finitely
generated $R$-modules, $\mathfrak{M}_R$ the maximal ideal of $M$, and let $M_{0}$ be a submodule of finite codimension. Then:
\begin{enumerate}[label=\roman*.)]
\item There is a number $\omega\in \mathbb{N}$ such that

\[
M_{0}\subset a(N)+\mathfrak{M}_{R}^{\omega}\cdot M_{0}
\]
implies

\[
M_{0}\subset a(N).
\]

\item If $\dim_{\mathbb{R}}M/a(N)<\infty$, then there exists a number $\omega'\in \mathbb{N}$
such that

\[
\mathfrak{M}_{R}^{\omega'}\cdot M\subset a(N).
\]

\end{enumerate}
\end{theorem}

Finally, we are able to prove the determinacy theorem for $\Gamma$-invariant
germs. But first we state a useful lemma, whose detailed proof we let for the technical results section.

\begin{lemma}\label{lemmadet}
If $g\in \mathcal{E}_n(\Gamma)$ has finite codimension, then there exists some $r\in \mathbb{N}$ such that for any $f\in\mathcal{E}_n(\Gamma)$ satisfying $j^rf=j^rg$, one has $\mathcal{M}_n^r(\Gamma)\subset T_{\mathcal{O}_f}$.
\end{lemma}

\begin{theorem}[Determinacy theorem]\label{theodet} Let $f\in\mathcal{E}_{n}(\Gamma)$.
then, the conditions are equivalent:
\begin{enumerate}[label=\roman*.)]
\item $f$ is $\Gamma$-finitely-determined,
\item $\mathcal{M}_{n}^{k}(\Gamma)\subset \tgsp_{\mathcal{O}_f}$ for some $k\in\mathbb{N}$,
\item $\cod_{\Gamma}(f)$ is finite. 
\end{enumerate}
\end{theorem}

\begin{proof}

$ii.)\implies iii.)$. This implication is trivial, since the dimension
of $\mathcal{M}_{n}(\Gamma)/\mathcal{M}_{n}^{k}(\Gamma)$ is always
finite for any $k$. Recall Definitions \ref{defbasic} and \ref{defcod}.

\smallskip

\noindent $iii.)\implies ii.)$. Let $a:\vec{\mathcal{E}}_{n}(\Gamma)\rightarrow\mathcal{E}_{n}(\Gamma)$ be the map given by $a(\phi)=\phi\cdot\vec{\nabla}f$. Now, we apply
Theorem \ref{mlemma} with $R=\mathcal{E}_{n}(\Gamma)$, $N=\vec{\mathcal{E}}_{n}(\Gamma)$
and $M=\mathcal{E}_{n}(\Gamma)$ (recall that $\vec{\mathcal{E}}_{n}(\Gamma)$ is finitely generated as an $\mathcal{E}_{n}(\Gamma)$-module). Since the maximal ideal of $\mathcal{E}_n(\Gamma)$ is $\mathcal{M}_n(\Gamma)$ and $a(\vec{\mathcal{E}}_{n}(\Gamma))=\tgsp_{\mathcal{O}_f}$ (see Definition (\ref{deftg})),
the implication directly follows from item $(b)$ of Theorem \ref{mlemma}.

\smallskip

\noindent $ii.)\implies i.)$.
Let $g\in\mathcal{E}_{n}(\Gamma)$ be such that
$j^rf=j^{r}g$, for the respective $r$ of Lemma \ref{lemmadet}. The idea of the proof is the following: let
\begin{equation}
F_{s}=sg(x)+(1-s)f(x),\quad s\in[0,1].
\end{equation}
To prove condition $i.)$, we must show that $F_0$ is equivalent to $F_1$. The goal will be then to prove that, for every $s\in [0,1]$, there exits an interval $I_s$ around $s$ such that for any $s'\in I_s$, $F_{s'}$ is $\Gamma$-equivalent to $F_s$. Since $\cup_{s\in [0,1]}I_s$ is an open cover for $[0,1]$ and $[0,1]$ is compact, it follows that there exists a finite sequence $(s_1,\dots,s_M)$ such that
\begin{equation}
F_0\sim F_{s_1}\sim \dots \sim F_{s_M}\sim F_1,
\end{equation}
which implies that $i.)$ holds. Hence, define for $s\in [0,1]$,
\begin{equation}
F_{s}(x,t)=(s+t)g(x)+(1-s-t)f(x).
\end{equation}
We must prove that there exists a
neighborhood of $0\in \mathbb{R}$ and a map $h(x,t)$ such that, if $t$ belongs
to that neighborhood, we have $h(\cdot,t)\in L_{x}(\Gamma)$ and $F_{s}(h(x,t),t)=F_{s}(x,0)$. Since $j^rF_{s}(\cdot,t)=j^rg$ for any $t$ and $g-f\in \mathcal{M}^r_n(\Gamma)$, by Lemma \ref{lemmadet}, it follows that there exists some $\vec{\xi}\in\vec{\mathcal{M}}_{n+1}$ with $\vec{\xi}(\cdot,t)\in\vec{\mathcal{M}}_n(\Gamma)$, such that for $t$ in some sufficiently small neighborhood $I$ of $0$,
\begin{equation}\label{eqsim}
\vec{\nabla}_xF_{s}(x,t)\cdot \vec{\xi}(x,t)+(g-f)(x)=0.
\end{equation}
Now, take $h(x,t)$ as the (unique) solution to:
\begin{equation}\label{eqdiff1}
\begin{cases}
\frac{\partial h(x,t)}{\partial t}=\vec{\xi}(h(x,t),t),\quad t\in I, \\
h(x,0)=x.
\end{cases}
\end{equation}
By the $\Gamma$-equivariance of $\vec{\xi}(\cdot,t)$ and the uniqueness of the solution of Equation (\ref{eqdiff1}), it follows that $h(\cdot,t)$ is also $\Gamma$-equivariant, and by Equation (\ref{eqsim}) one easily checks that

\begin{equation*}
\frac{\partial F_{s}(h(x,t),t)}{\partial t} = 0,\quad t\in I
\end{equation*}
i.e., $F_{s}(h(x,t),t)=F_{s}(x,0)$, for $t\in I$.

\bigskip

\noindent $i.)\implies ii.)$. Let $f$ be $\Gamma$-$k$-determined, and define
\begin{equation}
D_f=\{g\in \mathcal{E}_n(\Gamma)\;|\; j^kg=j^kf\}.
\end{equation}
$D_f$ is a linear subspace of $\mathcal{E}_n(\Gamma)$, and a tangent vector to $D_f$ at $f$ is of the type $f+th$, for any $h\in \mathcal{M}^{k+1}_n(\Gamma)$. Hence, it follows that its tangent space $T_{D_f}$ at $f$ is given by
\begin{equation}
T_{D_f}=\mathcal{M}^{k+1}_n(\Gamma).
\end{equation} 
But since $f$ is $\Gamma$-k-determined, it follows that $D_f$ is a subset of the orbit of $f$ with respect to $L_n(\Gamma)$; in particular, one also has the inclusion of the tangent spaces, i.e., $T_{D_f}\subset \tgsp_{\mathcal{O}_f}$.

\end{proof}

\begin{remark}
The determinacy theorem is of central importance, mainly because it
connects the size of the tangent space of $f$, with $f$ being $\Gamma$-finitely-determined. As a consequence, this theorem will allow us further to work in finite dimensional jet spaces for germs which are $\Gamma$-finitely-determined, and hence enabling to use important tools that comes from finite-dimensional differential geometry.
\end{remark}

\section{$\Gamma$-invariant unfoldings}

In this section we properly define the families of germs that will be relevant for the theory here proposed. In particular, we will be interested in families -- also called ufoldings -- whose elements are $\Gamma$-invariant. 

\begin{defn} A $\Gamma$-invariant $r$-unfolding of a germ $f\in\mathcal{E}_{n}(\Gamma)$
is a germ $F\in \mathcal{E}_{n+r}$ such that $F(x,0)=f(x)$ and, for some
neighborhood $W$ of $0\in\mathbb{R}^{r}$, one has $x\mapsto F(x,u)\in\mathcal{E}_n(\Gamma)$
for all $u\in W$. The set of all $\Gamma$-invariant $r$-unfoldings will be from now on denoted by $\mathcal{U}_{n,r}(\Gamma)$.

\end{defn}

The following definition is of central importance for the applications
of catastrophe theory that we are concerned, since it states what
classes of unfoldings we will call equivalent to each other. For our
purposes, we need equivalent unfoldings preserving the minima and
the critical points of our family while also preserving the assumed symmetry. Therefore, we allow for a $\Gamma$-equivariant diffeomorphism
that connects the variables, another that connects the parameters,
and a possibly third one which represents only a constant shift.

\begin{defn}\label{defeq2}
Let $f\in\mathcal{E}_{n}(\Gamma)$, $F$ be a
$\Gamma$-invariant $r$-unfolding of $f$ and $G$ be a 
$\Gamma$-invariant $d$-unfolding of $f$. We say that $G$ is \textit{induced
from} $F$ if there exists:
\begin{enumerate}[label=\roman*.)]
\item a germ $\phi\in \vec{\mathcal{M}}_{n+d,n}$, where $x\mapsto\phi(x,u)\in\vec{\mathcal{E}}_{n}(\Gamma)$
for $u$ in some neighborhood of $0\in\mathbb{R}^{d}$.
\item a germ $\psi\in \vec{\mathcal{M}}_{d,r}$,
\item a germ $h\in \mathcal{M}_d$,
\end{enumerate}
such that
\[
G(x,u)=F(\phi(x,u),\psi(u))+h(u).
\]
If $r=d$ and $\psi$ is a germ diffeomorphism, $F$ and $G$ are
said to be \textit{$\Gamma$-equivalent}. Moreover, if every $\Gamma$-invariant
unfolding of $f$ is induced from $F$, $F$ is said to be \textit{versal}.
\end{defn}

\begin{prop}\label{propinduced} Let $F_{1},F_{2},F_{3}$ be $\Gamma$-invariant $r_{1}$-,
$r_{2}$-, $r_{3}$-unfoldings, respectively, such that $F_{1}$ is
induced from $F_{2}$ and $F_{2}$ is induced from $F_{3}$. Then,
$F_{1}$ is induced from $F_{3}$.
\end{prop}

\begin{proof} Let $\phi_{i}$, $\psi_{i}$ and $h_{i}$, $i=1,2$,
be the corresponding germs of Definition \ref{defeq2}, satisfying,
for $u_{i}\in\mathbb{R}^{r_{i}}$,

\[
F_{i}(x,u_{i})=F_{i+1}(\phi_{i}(x,u_{i})),\psi_{i}(u_{i}))+h_{i}(u_{i}),\quad i=1,2.
\]
Defining

\begin{align*}
 & \phi_{3}(x,u_{1})=\phi_{2}(\phi_{1}(x,u_{1}),\psi_{1}(u_{1})),\quad\psi_{3}(u_{1})=\psi_{2}(\psi_{1}(u_{1})),\;\text{ and}\\
 & h_{3}(u_{1})=h_{1}(u_{1})+h_{2}(\psi(u_{1})),
\end{align*}
it follows that $\phi_{3},\psi_{3},h_{3}$ satisfy the conditions
(a), (b) and (c), of Definition \ref{defeq2}, with $d=r_{1}$
and $r=r_{3}$. Moreover, one has

\[
F_{1}(x,u_{1})=F_{3}(\phi_{3}(x,u_{1}),\psi_{3}(u_{1}))+h_{3}(u_{1}).
\]
Hence, $F_{1}$ is induced from $F_{3}$. \end{proof}

\begin{remark}
Note that, by the transitivity just stated above, it follows that the set of unfoldings that are $\Gamma$-equivalent to each other, according to Definition \ref{defeq2}, defines an equivalence class in $\mathcal{U}_{n,r}(\Gamma)$.
\end{remark}

\subsection{The tangent space generated by unfoldings}
We now look at the unfoldings of a germ $f\in\mathcal{M}_n(\Gamma)$ as a parametric family of germs in $\mathcal{M}_n(\Gamma)$, and analize its generated tangent space. Let $F\in \mathcal{U}_{n,r}(\Gamma)$ unfold $f\in\mathcal{M}_n(\Gamma)$, and Let $\mathbb{R}_\Gamma^n$ be the set of all points invariant under the action of $\Gamma$, i.e.,
\begin{equation*}
\mathbb{R}^n_\Gamma=\{x\in \mathbb{R}^n\;|\; \gamma\cdot x=x\text{ for all }\gamma\in \Gamma\}.
\end{equation*}
Consider the mapping $S_F:\mathbb{R}^n_\Gamma\times \mathbb{R}^r\rightarrow \mathcal{M}_n(\Gamma)$ defined by
\begin{equation*}
S_F(x,u)(x')\doteq F(x+x',u)-F(x,u). 	
\end{equation*}
Let $(v_1,v_2)$ be arbitrary vectors in $\mathbb{R}^n_\Gamma\times \mathbb{R}^r$. Then, the tangent vector of the differential of $S_F$ with respect to the curve $(\eta_1(t),\eta_2(t))=(tv_1,tv_2)$ is given by
\begin{equation*}
\frac{\partial S_F(\eta_1(t),\eta_2(t))(x')}{\partial t}|_{t=0} =v_1\cdot (\vec{\nabla}f(x')-\vec{\nabla}f(0)) +v_2\cdot( \vec{\nabla}_uF(x',0)- \vec{\nabla}_uF(0,0)).
\end{equation*}
Hence, the tangent space associated with the differential of $S_F$ is then given by
\begin{equation*}
\{v_1\cdot (\vec{\nabla}f(x)-\vec{\nabla}f(0)) +v_2\cdot( \vec{\nabla}_uF(x,0)- \vec{\nabla}_uF(0,0))\;|\; (v_1,v_2)\in \mathbb{R}^n_\Gamma\times \mathbb{R}^r\}\in \mathcal{M}_n(\Gamma).
\end{equation*}
\begin{defn}
Let $F\in\mathcal{U}_{n,r}(\Gamma)$ be a $\Gamma$-invariant unfolding of some $f\in\mathcal{M}_n(\Gamma)$.
\begin{equation}
T_{S_F}\doteq \{v_1\cdot (\vec{\nabla}f(x)-\vec{\nabla}f(0)) +v_2\cdot( \vec{\nabla}_uF(x,0)- \vec{\nabla}_uF(0,0))\;|\; (v_1,v_2)\in \mathbb{R}^n_\Gamma\times \mathbb{R}^r\},
\end{equation}
where
\begin{equation}
\mathbb{R}^n_\Gamma=\{x\in \mathbb{R}^n\;|\; \gamma\cdot x=x\text{ for all }\gamma\in \Gamma\}.
\end{equation}
Moreover, we say that $F$ is \textit{$\Gamma$-transversal} if $T_{S_F}$ is transversal to $T_{\mathcal{O}_f}$ in $\mathcal{M}_n(\Gamma)$, i.e., if
\begin{equation}\label{eqtransv1}
\mathcal{M}_n(\Gamma)=T_{\mathcal{O}_f}+T_{S_F}.
\end{equation}
\end{defn}
We now introduce the following notation:
\begin{defn}
Let $R$ be a ring and $M$ an $R$-module. Given $v_1,\dots,v_n\in M$, We denote by
\begin{equation}
\langle v_1,\dots ,v_n\rangle_R
\end{equation}
the $R$-sub-module of $M$ generated by  $v_1,\dots,v_n$.
\end{defn}

\begin{prop}\label{proptransversal}
Let $F\in\mathcal{U}_{n,r}(\Gamma)$ be a $\Gamma$-invariant unfolding of some $f\in\mathcal{M}_n(\Gamma)$. Then, $F$ is $\Gamma$-transversal if and only if 
\begin{equation}\label{eqtransv2}
\mathcal{E}_n(\Gamma) = \tgsp_{\mathcal{O}_f}+ \langle1,\alpha_1(F),\dots,\alpha_r(F)\rangle_\mathbb{R},
\end{equation}
where
\begin{equation}
\alpha_i(F)\doteq\frac{\partial F(x,0)}{\partial u_i}.
\end{equation}

\end{prop}

\begin{proof}
Note that $\mathcal{E}_n(\Gamma)\cong\mathcal{M}_n(\Gamma)+\mathbb{R}$, and
\begin{equation}
\{v\cdot \vec{\nabla}f\;|\; v\in \mathbb{R}^n_\Gamma\}+T_{\mathcal{O}_f}=\{\phi\cdot\vec{\nabla}f\;|\; \phi\in \vec{\mathcal{E}}_{n}(\Gamma)\}=\tgsp_{\mathcal{O}_f}.
\end{equation}
Hence, adding $\mathbb{R}$ to both sides of Equation (\ref{eqtransv1}), since we have the inclusions $\{v\cdot \vec{\nabla}f(0)\;|\; v\in \mathbb{R}^n_\Gamma\}\subset \mathbb{R}$ and $\{v\cdot \vec{\nabla}_uF(0,0)\;|\; v\in \mathbb{R}^n\}\subset \mathbb{R}$, it is not hard to see that one may arrive at
\begin{equation}
\mathcal{E}_n(\Gamma) = \tgsp_{\mathcal{O}_f}+ \langle1,\alpha_1(F),\dots,\alpha_r(F)\rangle_\mathbb{R}.
\end{equation}

\end{proof}

\subsection{Transversality theorem for $\Gamma$-invariant unfoldings}

Here we present the main result of the section: namely the transversality theorem for symmetric unfoldings. Such theorem is of fundamental importance, since it connects the
property of an unfolding being versal, i.e., of inducing any other
unfolding of the same germ, to the algebraic condition related to the transversality (see Proposition \ref{proptransversal}), which in many cases can be
analytically checked.

\begin{defn}
Let $F,G\in\mathcal{U}_{n,r}(\Gamma)$ be two $\Gamma$-transversal
unfoldings of the same germ $f\in\mathcal{M}_n(\Gamma)$. $F$ and $G$ are said to be
\textit{elementary $\Gamma$-homotopic} when, for all $s\in[0,1]$, the unfolding

\[
H_{s}=sG+(1-s)F
\]
is $\Gamma$-transversal. Moreover, $F$ and $G$ are said to be \textit{$\Gamma$-homotopic}
when there exists a finite sequence of $\Gamma$-invariant unfoldings $F=F_{0},F_{1},\dots,F_{n-1},F_{n}=G$
such that $F_{i-1}$ and $F_{i}$ are $\Gamma$-elementary homotopic for all
$i=1,\dots,n$.
\end{defn}

\begin{theorem}\label{eqtransversal}
Let $F,G\in \mathcal{U}_{n,r}(\Gamma)$ be two $\Gamma$-transversal
$r$-unfoldings of the same germ $\eta\in\mathcal{M}_n(\Gamma)$. Then, $F$ and
$G$ are $\Gamma$-equivalent.
\end{theorem}

\begin{proof}
Recall Definition \ref{defeq2}. First, it shall be shown that it is enough to prove
the theorem for the special case where $F$ and $G$ are $\Gamma$-elementary homotopic. Let $k\in \mathbb{N}$ be such that $\mathcal{M}^k_n(\Gamma)\subset \tgsp_{\mathcal{O}_f}$ (such $k$ always exists by Theorem \ref{theodet}). Define in $J_n^{k}(\Gamma)$ the subspace
\begin{equation}
A=j_{k}\left(\tgsp_{\mathcal{O}_f}\right).
\end{equation}
Note that $J_{n}^{k}(\Gamma)$ is isomorphic to an Euclidean space. Hence,
let $A^{\perp}$ be the orthogonal subspace of $A$, and let $\{z_{1},\dots,z_{q}\}$
be a basis of $A^{\perp}$ (later we discuss the case of $A^{\perp}=\{0\}$). Since $F$ is $\Gamma$-transversal by hypothesis,
there exists some $a\in A$ and some constants $a_{1},\dots,a_{r}\in\mathbb{R}$
such that
\begin{equation}
z_{1}=a+\sum_{j=1}^{r}a_{j}j^{k}(\alpha_{j}(F)).\label{eqp1}
\end{equation}
Let $j_{1}$ be such that $a_{j_{1}}\neq0$
in the above equation. Hence, define a new unfolding $F_{1}$ given
by
\[
F_{1}(x,u)=f(x)+u_{1}\alpha_{1}(F)+\dots+u_{j_{1}}z_{1}+\dots+u_{r}\alpha_{r}(F),
\]
if $a_{j_{1}}>0$, or
\[
F_{1}(x,u)=f(x)+u_{1}\alpha_{1}(F)+\dots-u_{j_{1}}z_{1}+\dots+u_{r}\alpha_{r}(F),
\]
if $a_{j_{1}}<0$. Note that $F_{1}$ is also a $\Gamma$-invariant $r$-unfolding
of $f$. Moreover, $F$ and $F_{1}$ are $\Gamma$-elementary homotopic. Indeed,
for an unfolding $H_{s}=sF_{1}+(1-s)F$ where $s\in[0,1]$, $j^{k}(\alpha_{j}(H_{s}))=j^{k}(\alpha_{j}(F))$
for all $j\neq j_{1}$, and
\[
j^k(\alpha_{j_{1}}(H_{s}))=sj^{k}(\alpha_{j_{1}}(F))+sgn(a_{j_{1}})(1-s)z_{1}.
\]
By Equation (\ref{eqp1}), $j^{k}(\alpha_{j_{1}}(F))$ can be expressed
as
\begin{align*}
j^{k}(\alpha_{j_{1}}(F))= & \frac{j^{k}(\alpha_{j_{1}}(H_{s}))-sgn(a_{j_{1}})\left(a+\sum_{\substack{j=1\\
j\neq j_{1}
}
}^{r}a_{j_{1}}j_{k}(\alpha_{j}(F))\right)}{s+(1-s)\vert a_{j_{1}}\vert}\\
= & \frac{j^{k}(\alpha_{j_{1}}(H_{s}))-sgn(a_{j_{1}})\left(a+\sum_{\substack{j=1\\
j\neq j_{1}
}
}^{r}a_{j_{1}}j^{k}(\alpha_{j}(H_{s}))\right)}{s+(1-s)\vert a_{j_{1}}\vert}.
\end{align*}
Hence, since $F$ is $\Gamma$-transversal, it is easy to see that that $H_{s}$
will also be $\Gamma$-transversal. Moreover, by a diffeomorphic change of
parameters $e\in\epsilon(r,r)$ given by

\[
e(u_{1},\dots,u_{j_{1}},\dots,u_{r})=(u_{j_{1}},u_{2},\dots,sgn(a_{j_{1}})u_{1},u_{j_{1}+1},\dots,u_{r}),
\]
one can obtain another $\Gamma$-transversal $r$-unfolding $\tilde{F}_{1}$
of $f$ given by
\begin{align*}
\tilde{F}_{1}(x,u)=F_{1}(x,e(u))= & f(x)+u_{1}z_{1}+u_{2}\alpha_{2}(F)+\dots+u_{j_{1}}\alpha_{1}(F)+u_{j_{1}+1}\alpha_{j_{1}+1}(F)\\
 & +\dots+u_{r}\alpha_{r}(F),
\end{align*}
where $\tilde{F}_{1}$ is clearly $\Gamma$-equivalent to $F_{1}$. Now, consider
a $\Gamma$-transversal $r$-unfolding $F_{m}$ of $f$, where $m<q$, such
that $\alpha_{j}(F_{m})=z_{j}$ for $j\leq m$, and $\alpha_{j}(F_{m})=\alpha_{j}(F)$
for $j>m$. Then, since $F_{m}$ is $\Gamma$-transversal, it follows that,
similar to the above case, there exists some $a\in A$ and some constants
$a_{1},\dots,a_{r}\in\mathbb{R}$ such that

\[
z_{m+1}=a+\sum_{j=1}^{m}a_{j}z_{m}+\sum_{j=m+1}^{r}a_{j}j^{k}(\alpha_{j}(F)).
\]
Again, clearly $a_{j'}\neq0$ for some $j'>m$. Defining a new unfolding
$F_{m+1}$ of $f$ by
\begin{align*}
F_{m+1}(x,u)= & f(x)+u_{1}z_{1}+\dots+u_{m}z_{m}+u_{m+1}\alpha_{m+1}(F_{m})+\dots+sgn(a_{j'})u_{j'}z_{m+1}\\
 & +u_{j'+1}\alpha_{j'+1}(F_{m})+\dots+u_{r}\alpha_{r}(F_{m}),
\end{align*}
and proceeding in a similar way to the above discussion, it is not
hard to see that $F_{m+1}$ is also $\Gamma$-transversal and elementary $\Gamma$-homotopic
to $F_{m}$. Moreover, as before, by the germ diffeomorphism $e\in\epsilon(r,r)$
\[
e(u_{1},\dots,u_{r})=(u_{1},\dots,u_{m},u_{j'},u_{m+2},\dots,u_{j'-1},sgn(a_{j'})u_{m+1},u_{j'+1},\dots,u_{r}),
\]
one can find another unfolding $\tilde{F}_{m+1}$ of $f$ given by
\begin{align*}
\tilde{F}_{m+1}(x,u)=F_{m+1}(x,e(u))= & f(x)+u_{1}z_{1}+\dots+u_{m+1}z_{m+1}+u_{m+2}\alpha_{m+2}(F_{m})\\
+\dots+ & u_{j'}\alpha_{m+1}(F_{m})+u_{j'+1}\alpha_{j'+1}(F_{m})+\dots+u_{r}\alpha_{r}(F_{m}),
\end{align*}
where $\tilde{F}_{m+1}$ is $\Gamma$-equivalent to $F_{m+1}$. By induction,
it follows that there exists a finite sequence of $\Gamma$-transversal unfoldings
of $f$, $F,F_{1},\tilde{F}_{q},\dots,F_{q},\tilde{F}_{n}$, where
every unfolding in the sequence is either $\Gamma$-equivalent or elementary
$\Gamma$-homotopic to the next
, and $\tilde{F}_{q}$ is given by
\[
\tilde{F}_{n}(x,u)=f(x)+u_{1}z_{1}+\dots+u_{q}z_{q}+h(x,u_{q+1},\dots,u_{r}).
\]
By a straightforward calculation it also follows that $\tilde{F}_{q}$
is elementary $\Gamma$-homotopic to the unfolding
\[
Y(x,u)=f(x)+u_{1}z_{1}+\dots+u_{q}z_{q}.
\]
Hence, for any $\Gamma$-transversal unfolding $F,G$ of $f$ with $r$ parameters,
there exists a finite sequence of unfoldings $F,F_{1},\tilde{F}_{1},\dots,F_{q},\tilde{F}_{q},Y,\tilde{G}_{q},G_{q},\dots,\tilde{G}_{1},G_{1},G$
such that every unfolding in the sequence is either $\Gamma$-equivalent or
elementary $\Gamma$-homotopic to the next
. Consequently, If one proves that elementary $\Gamma$-homotopic unfoldings are also
$\Gamma$-equivalent, one proves that $F$ is $\Gamma$-equivalent to $G$.

\begin{remark}
If $A^{\perp}=\{0\}$ i.e., $\cod_\Gamma(f)=0$, then we take $Y$ to be the trivial unfolding $Y(x,u)=f(x)$.
\end{remark}

The idea of the proof is to show that, when $F$ and $G$ are elementary
$\Gamma$-homotopic, then for every $s\in[0,1]$, there is an open interval
$I_{s}$ of $[0,1]$ with $s\in I_{s}$ such that for all $s'\in I_{s}$,
the unfolding
\[
H_{s'}=s'G+(1-s')F
\]
is isomorphic to $H_{s}$. Assuming that this holds, then the union
$\cup_{s\in[0,1]}I_{s}$ is an open cover of $[0,1]$. But since $[0,1]$
is compact, it follows that one may find a finite sequence $\{s_{0},s_{1},\dots,s_{M}\}$
with $s_{0}=0$ and $s_{M}=1$ such that $H_{0},H_{s_{1}},\dots,H_{s_{M-1}},H_{1}$
are all equivalent. But since $H_{0}=F$ and $H_{1}=G$, this implies
that $F$ and $G$ are equivalent. Hence, let $F,G$ be two elementary
homotopic even unfoldings of $f$, and $s\in(0,1)$. Define
\[
H(x,u,t)=(s+t)G+(1-s-t)F.
\]
The goal is to show that for some sufficiently small $t\in\mathbb{R}$,
$H(x,u,t)$ is equivalent to $H(x,u,0)$. Note that there exits a
neighborhood $V$ of $0\in\mathbb{R}$ such that if $t\in V$, $(x,u)\mapsto H(x,u,t)$
is $\Gamma$-transversal, since by hypothesis $F$ and $G$ are elementary
 $\Gamma$-homotopic. Therefore, Corollary \ref{coroalg} applies, and it follows that
\begin{equation}\label{eqtil}
\mathcal{U}_{n,r+1}(\Gamma)=\tilde{T}_H+\left\langle 1,\frac{\partial H}{\partial u_1},\dots,\frac{\partial H}{\partial u_r}\right\rangle_{\mathcal{E}_{r}}.
\end{equation}
Now, note that $H(x,u,t)\in\mathcal{U}_{n,r}(\Gamma)$ and 
\begin{equation}
\frac{\partial H(x,0,t)}{\partial t}=\frac{\partial(f(x))}{\partial t}=0.
\end{equation}
Hence, by lemma \ref{lemma1}, 
\begin{equation}
\frac{\partial H}{\partial t}\in \mathcal{M}_r\cdot\mathcal{U}_{n,r+1}(\Gamma).
\end{equation}
Multiplying Equation (\ref{eqtil}) by $\mathcal{M}_r$, it follows that there exists some germs:
\begin{align*}
&\xi(x,u,t)\in \mathcal{M}_r\cdot\vec{\mathcal{E}}_{n+r+1,n},\; \text{ with }x\mapsto \xi(x,u,t)\in \vec{\mathcal{E}}_n(\Gamma) \text{ for $(u,t)$ in some neighborhood of $0\in \mathbb{R}^{r+1}$},\\
&\chi(u,t)\in \mathcal{M}_r\cdot\vec{\mathcal{E}}_{r+1,r},\;\text{ and}\\
&\chi_0(u,t)\in \mathcal{M}_r\cdot\mathcal{E}_{r+1}.
\end{align*}
Such that
\begin{equation}\label{eq11}
\frac{\partial H}{\partial t}(x,u,t)=\xi(x,u,t)\cdot \vec{\nabla}_xH(x,u,t) + \chi(u,t)\cdot \vec{\nabla}_uH(x,u,t)+\chi_{0}(u,t).
\end{equation}
Applying the geometric lemma (Lemma \ref{difeq2}) for $H\in\mathcal{U}_{n,r+1}(\Gamma)$, it follows that there exists some germs
\begin{align*}
&\phi(x,u,t)\in \vec{\mathcal{E}}_{n+r+1,n},\; \text{ with }x\mapsto \xi(x,u,t)\in \vec{\mathcal{E}}_n(\Gamma) \text{ for $(u,t)$ in some neighborhood of $0\in \mathbb{R}^{r+1}$},\\
&\psi(u,t)\in \vec{\mathcal{E}}_{r+1,r},\;\text{ and}\\
& h(u,t)\in \mathcal{E}_{r+1}
\end{align*}
such that
\begin{equation}\label{eqequiv}
H(x,u,0)=H(\phi(x,u,t),\psi(u,t),t)+h(u,t)
\end{equation}
for small enough $t\in \mathbb{R}$. Note also that, by definition of the germs $\phi, \psi$ and $h$, they in particular satisfy the following system of O.D.Es. (for more details see the proof of Lemma \ref{difeq2})
\begin{equation}
\begin{cases}
\frac{\partial \phi (x,0,t)}{\partial t}=-\xi(\phi(x,0,t),0,t),\quad \phi(x,0,0)=x, \\
\frac{\partial\psi(0,t)}{\partial t}=-\chi(\psi(0,t),t),\quad\quad\quad\:\psi(0,0)=0,\\
\frac{\partial h(0,t)}{\partial t}=-\chi_0(0,t),\quad\quad\quad\quad\quad \,h(0,0)=0,
\end{cases}
\end{equation}
for small enough $t\in\mathbb{R}$. Since $\xi(x,0,t)=\chi(0,t)=\chi_0(0,t)=0$, the unique solution of the above system of O.D.Es. is
\begin{equation}
\phi(x,0,t)=x,\;\; \psi(0,t)=0,\;\; h(0,t)=0.
\end{equation}
In particular, this and <equation (\ref{eq11}) together imply that $H(x,u,t)$ is equivalent to $H(x,u,0) $ for small enough $t\in\mathbb{R}$.
\end{proof}

\begin{theorem}[Transversality theorem for $\Gamma$-invariant unfoldings]\label{trasnvertheo}
Let $F\in\mathcal{U}_{n,r}(\Gamma)$ be a $\Gamma$-equivariant $r$-unfolding
of $f\in\mathcal{M}_{x}(\Gamma)$. Then, the following are equivalent:
\begin{enumerate}[label=\roman*.)]
\item $F$ is versal,
\item $F$ is $\Gamma$-transversal, i.e., $\mathcal{E}_{x}(\Gamma)=\tgsp_{\mathcal{O}_f}+\left\langle 1,\alpha_{1}(F),\dots,\alpha_{r}(F)\right\rangle _{\mathbb{R}}$. 
\end{enumerate}
\end{theorem}

\begin{proof} $i)\implies ii)$. Let $p\in\mathcal{E}_{x}(\Gamma)$
be an arbitrary $\Gamma$-invariant germ, and consider the $1$-unfolding $G(x,t)=f(x)+tp(x)$.
Clearly $G$ is a $\Gamma$-invariant unfolding of $f$, and if
$F$ is versal, then by hypothesis there exists some germs $\phi,\psi$
and $h$ such that
\[
G(x,t)=f(x)+tp(x)=F(\phi(x,t),\psi(t))+h(t).
\]
Differentiating with respect to $t$ and taking $t=0$, we arrive
at
\[
p(x)=\frac{\partial \phi(x,0)}{\partial t}\cdot \vec{\nabla}_xf(x)+\sum_{i}\frac{\partial F(x,0)}{\partial u_{i}}\frac{\partial\psi_{i}(0)}{\partial t}+\frac{\partial h(0)}{\partial t}.
\]
By the assumed $\Gamma$-equivariance of $t\mapsto\phi(x,t)$ in a suitable
neighborhood of $0\in\mathbb{R}$, one has
\begin{equation}
\frac{\partial \phi(x,0)}{\partial t}\cdot\vec{\nabla}_xf(x)\in T_{\mathcal{O}_f},
\end{equation}
In particular,
\begin{equation}
p\in \tgsp_{\mathcal{O}_f}+\left\langle 1,\alpha_{1}(F),\dots,\alpha_{r}(F)\right\rangle _{\mathbb{R}},
\end{equation}
hence $ii)$ follows.

\smallskip

\noindent$ii)\implies i)$. Let $F$ be a $\Gamma$-transversal $r$-unfolding of $f$, and let $G$ be an arbitrary $\Gamma$-invariant $d$-unfolding of $f$. Define two $(r+d)$-unfoldings of $f$ by
\begin{equation*}
H(x,u,v)=F(x,u)+G(x,v)-f(x), \;\text{and }\; K(x,u,v)=F(x,u).
\end{equation*}
Clearly, $H$ and $K$ are $\Gamma$-transversal, as $F$ is $\Gamma$-transversal. Hence, by Proposition \ref{eqtransversal}, $H$ and $K$ are $\Gamma$-equivalent. Moreover, $G$ is induced from $H$, by the germs $\phi(x,u)=x$, $\psi(u,v)=u$, $\gamma(u,v)=0$, and $K$ is induced from $F$, by the germs $\phi(x,u)=x$, $\psi(v)=(0,v)$, $\gamma(v)=0$. Therefore, by Proposition \ref{propinduced}, $G$ is induced from $F$, i.e., $F$ is versal.

\end{proof}

\begin{remark}
In view of the definition of $\Gamma$-codimension (Definition \ref{defcod}) and the algebraic condition of transversality (Equation (\ref{eqtransv2})), it clearly follows that the minimum possible number of parameters that a $\Gamma$-invariant $r$-unfolding $F$ of a germ $f$ needs to have in order to be $\Gamma$-transversal is $r=\cod_\Gamma(f)$. By thetTransversality theorem, one may also state that this is also the minimum possible number of parameters for $F$ to be versal. Such unfoldings which are versal/transversal while having the minimum number of parameters $r=\cod_\Gamma(f)$ are usually referred to in the literature as \textit{universal unfoldings}.
\end{remark}

\subsection{Stability theorem for $\Gamma$-invariant unfoldings}

We now move on to the proof of the stability theorem. In short, the theorem states that any perturbation added to a transversal unfolding, if small enough (in a suitable topology), does not break its transversality, and in fact the perturbed unfolding is equivalent to the original unfolding. For the case of $\Gamma$-invariant unfoldings, we must be cautious with the kind of perturbations that we allow. In general, transversal $\Gamma$-invariant unfoldings are \textbf{not} stable when the perturbation breaks the assumed symmetry of the unfolding. Therefore, here we analize the situation for symmetry-preserving perturbations. First, we start by defining a suitable topology for the perturbations.

\begin{defn} \label{deftop}
Let $U$ be a neighborhood of $0\in \mathbb{R}^n\times \mathbb{R}^r$.
Define $C_{\Gamma}^{\infty}(U,\mathbb{R})$ as the set of all infinitely
differentiable functions $F:U\rightarrow\mathbb{R}$ such that $x\mapsto F(x,u)$
is $\Gamma$-invariant for all $u$ such that $(0,u)\in U$. The topology considered
in $C_{\Gamma}^{\infty}(U,\mathbb{R})$ is the so-called \textit{$C^\infty$-topology}, whose local basis at $0$ is the union of all sets
\[
V(L,k,\epsilon)=\left\{ F\in C_{\Gamma}^{\infty}(U,\mathbb{R})\;|\;\sum_{|\alpha|\leq k}\frac{1}{\alpha!}\sup_{(x,u)\in L}\left\vert \frac{\partial^{|\alpha|}(F)(x,u)}{\partial x_{1}^{\alpha_{1}}\cdots\partial x_{n}^{\alpha_{n}}}\right\vert <\epsilon\right\} ,
\]
where $L\subset U$ is compact, $k\in\mathbb{N}$
and $\epsilon>0$.
\end{defn}

\begin{remark}
Note that in such topology, by definition, for any $k\in\mathbb{N}$, there exists a neighborhood $V$ of $0$ such that the derivatives of all of the elements in $V$ up to order $k$ are as small as one wants. This property will be very important in the proof of the stability theorem, as we will need to control all the derivatives of the perturbations up to some finite order $k$. 
\end{remark}

An arbitrary perturbation added to an unfolding may also have the effect of dislocating the point of singularity of the initial germ. Therefore, We need to modify our equivalence relation, to include also diffeomorphisms that does not preserve the origin.

\begin{defn} Let $U,V$ be open subsets of $\mathbb{R}^{n+r}$,
$G:U\rightarrow\mathbb{R}$, $F:V\rightarrow\mathbb{R}$ two smooth
functions, and $x,y\in\mathbb{R}^n$ such that
\begin{equation}
x=\gamma x\text{ for all }\gamma\in \Gamma,\quad\text{and}\quad y=\gamma y\text{ for all }\gamma\in \Gamma.
\end{equation}
Then, \textit{$G$ at $(x,u)\in U$ is equivalent to $F$
at $(y,v)\in V$} if, for some neighborhood $U_{1}$ of $\mathbb{R}^{n}$
and some neighborhood $U_{2}$ of $\mathbb{R}^{r}$, such that $U_{1}\times U_{2}\subset U$,
there exist smooth functions

\begin{enumerate}[label=\roman*.)]
\item $\phi:U_1\times U_2\rightarrow\mathbb{R}^n$, such that $\phi(0,0)=x'$ and $x\mapsto\phi(x,u)$ is a diffeomorphism for all $u\in U_2$,
\item $\psi:U_2\rightarrow \mathbb{R}^r$, such that $\psi(0)=u'$ and $\psi$ is a diffeomprhism,
\item $h:U_2\rightarrow\mathbb{R}$,
\end{enumerate}
such that, for all $(x',u')\in U_{1}\times U_{2}$, one has
\[
G(x',u')=F(\phi(x',u'),\psi(u'))+\gamma(u').
\]
\end{defn}

\begin{defn}\label{defstable}
Let $F\in \mathcal{U}_{n,r}(\Gamma)$, with representative $\tilde{F}$. $F$ is said to be \textit{stable} if for every open neighborhood $U$ of $0\in \mathbb{R}^{n+r}$, there exists a neighborhood $W$ of $\tilde{F}$ in the $C^\infty$-topology of $C^\infty_\Gamma(\mathbb{R}^{n+r},\mathbb{R})$ such that for every $\tilde{F}'\in W$, there exists some $(x',u')\in U^\Gamma$, where $U^\Gamma=\{(x,u)\in U \;|\; x=\gamma x\text{ for all }\gamma\in \Gamma\}$, such that $F$ is equivalent to $\tilde{F}'$ at $(x',u')$; i.e., there exists some neighbhood $U_1$ of $0\in \mathbb{R}^n$, some neighborhood $U_2$ of $0\in \mathbb{R}^r$, and some smooth functions

\begin{enumerate}[label=\roman*.)]
\item $\phi:U_1\times U_2\rightarrow\mathbb{R}^n$, such that $\phi(0,0)=x'$ and $x\mapsto\phi(x,u)$ is a diffeomorphism for all $u\in U_2$,
\item $\psi:U_2\rightarrow \mathbb{R}^r$, such that $\psi(0)=u'$ and $\psi$ is a diffeomprhism,
\item $h:U_2\rightarrow\mathbb{R}$,
\end{enumerate}
such that

\begin{equation*}
G(x,u)=\tilde{G}'(\phi(x,u),\psi(u))+h(u).
\end{equation*}
\end{defn}

\begin{defn} Let $U=V\times W\subset\mathbb{R}^{n+r}$ be an open
neighborhood of $0$ and let $F\in C^\infty_\Gamma(U,\mathbb{R})$. For $k\in\mathbb{N}$, define $j_{1}^{k}F:\mathbb{R}^n_\Gamma\times \mathbb{R}^r\rightarrow J_{n}^{k}(\Gamma)$
by
\[
j_{1}^{k}F(x,u)(z)=\sum_{|\alpha|\leq k,\alpha\neq0}\frac{1}{\alpha!}\frac{\partial^{|\alpha|}F(x,u)}{\partial x_{1}^{\alpha_{1}}\cdots\partial x_{n}^{\alpha_{n}}}z_{1}^{\alpha_{1}}\dots z_{n}^{\alpha_{n}},
\]
i.e., $j_{1}^{k}F(x,u)(z)=j^{k}(F(x+z,u))-F(x,u)$. \end{defn}
\begin{remark}
We recall that
\begin{equation}
\mathbb{R}_\Gamma^n=\{x\in \mathbb{R}^n\;|\; x=\gamma x\text{ for all }\gamma\in \Gamma\}.
\end{equation}
\end{remark}

\begin{theorem}[Stability theorem for $\Gamma$-invariant unfoldings]\label{theostable} Let $f\in \mathcal{M}_{n}(\Gamma)$, and let $F\in\mathcal{U}_{n,r}(\Gamma)$
be a $\Gamma$-transversal  unfolding of $f$. Then, $F$ is $\Gamma$-stable.
\end{theorem}

\begin{proof} Let $V,W$ be an open neighborhood of $0\in\mathbb{R}^{n}$,
$\mathbb{R}^{r}$, respectively, and $F'$ a representative of $F$
defined on $U=V\times W$. For $g\in C^{\infty}(U,\mathbb{R})$
and $w\in W$, define $\Xi:W\times C^{\infty}(U,\mathbb{R})\rightarrow C^{\infty}(U,\mathbb{R})$ by
\[
\Xi(w,g)(x,u)=g(x,w+u)-g(0,w).
\]
Define also a function $\zeta:W\times C^{\infty}(U,\mathbb{R})\rightarrow C^{\infty}(U,\mathbb{R})$
by
\[
\zeta(w,g)(x)=g(x,w)-g(0,w)=\Xi(w,g)(x,0).
\]
Now, note that if $G\in\mathcal{U}_{n,r}(\Gamma)$ is a $\Gamma$-invariant unfolding of a germ $\mu\in \mathcal{M}_{n}(\Gamma)$ that is $\Gamma$-finitely-determinate, then it follows that
$G$ is $\Gamma$-transversal if and only if $J_{n}^{k}(\Gamma)$ is generated
by the elements: $1,j^{k}(\alpha_{1}(G)),\dots,j^{k}(\alpha_{r}(G))$
and by a basis of $j^{k}(\tgsp_{\mathcal{O}_\eta})$, for some $k\in\mathbb{N}$ coming from item $ii.)$ of Theorem \ref{theodet}. Denote every element in the above list, when taking the unfolding $F$, by $\beta_{1}(F),\dots,\beta_{p}(F)$,
in a way such that $j^{k}(\beta_{1}(F)),\dots,j^{k}(\beta_{q}(F))$
$(q\leq p)$ form a basis of $J_{n}^{k}(\Gamma)$. Then, choose a compact
neighborhood $K$ of $0\in\mathbb{R}^{r}$ such that for all $w\in K$:
\begin{enumerate}[label=\roman*.)]
\item $j^{k}\zeta(w,F')$ $\Gamma$-$k$-determinate, and 
\item $j^{k}\beta_{1}(\Xi(w,F')),\dots,j^{k}\beta_{q}(\Xi(w,F'))$ form
a basis of $J_{n}^{k}(\Gamma)$. 
\end{enumerate}
Let $c>0$ be such that $[-c,c]^{r}\subset K$, and define $S\subset\mathbb{R}^{q}$
as the $q-1$ sphere of radius $c$ centered at $0$. Then, by the
continuity of the mapping $\mathfrak{i}:K\times S\rightarrow J_{n}^{k}(\Gamma)$
given by
\[
\mathfrak{i}(w,s_{1},\dots,s_{q})=\sum_{i=1}^{q}s_{i}j^{k}\beta_{i}(\Xi(w,F')),
\]
and the compactness of $K\times S$, it follows that there exists some $\delta>0$ such that
\[
\left\Vert \sum_{i=1}^{q}s_{i}j^{k}\beta_{i}(\Xi(w,F'))\right\Vert \geq\delta
\]
for all $w\in K$ and $(s_{1},\dots,s_{q})\in S$. Moreover, for any
$h\in C_{\Gamma}^{\infty}(U,\mathbb{R})$ with derivatives up to order
$k+1$ small enough in $\{0\}\times K$, it also follows that
\begin{align}\label{eqskdet}
 & \left\Vert \sum_{i=1}^{q}s_{i}j^{k}\beta_{i}(\Xi(w,F'+th))\right\Vert \geq\frac{\delta}{2}
\end{align}
for all $w\in W$, $(s_{1},\dots,s_{q})\in S$ and $t\in[0,1]$.
For $h\in C_{\Gamma}^{\infty}(U,\mathbb{R})$ with derivatives up to order
$k+1$ small enough in $\{0\}\times K$, one also has
\begin{equation}
\Vert j^{k}h(\cdot,w)\Vert<\frac{\delta}{2}\;\text{ for all }w\in K.\label{eqdelta2}
\end{equation}
Hence, suppose that $h$ is in a convenient neighborhood of $0$ in
$C_{\Gamma}^{\infty}(U,\mathbb{R})$ such that Equations (\ref{eqskdet}) and
(\ref{eqdelta2}) are satisfied. For $w\in K$ and $a\in[0,1]$, define
$H\in\mathcal{U}_{n,r+1}(\Gamma)$ by

\[
H(x,u,t)=(F'+(a+t)h)(x,w+u)-(F'+ah)(0,w).
\]
The strategy of the proof is very similar to the one of Theorem \ref{eqtransversal}, i.e., we will prove that $H(x,u,0)$ is $\Gamma$-equivalent to $H(x,u,t)$ for $t$ in some neighborhood of $0\in \mathbb{R}$, and this concludes the theorem. Note that $H(x,0,0)=f=\zeta(w,F'+ah)$ and $H(x,u,0)=\Xi(w,F'+ah)$.
Therefore, since Equation (\ref{eqskdet}) implies that the set $\{j^{k}\beta_{1}(\Xi(w,F'+ah)),\dots,j^{k}\beta_{q}(\Xi(w,F'+ah))\}$
is linearly iindependent (and hence forms a basis of $J_{n}^{k}(\Gamma)$), it follows that
$H$ is a $\Gamma$-transversal unfolding of $f$. in particular, one has, by Corollary \ref{coroalg}
\begin{equation}\label{eq59}
\mathcal{U}_{n,r+1}(\Gamma)=\tilde{T}_H+\left\langle 1,\frac{\partial H}{\partial u_1},\dots,\frac{\partial H}{\partial u_r}\right\rangle_{\mathcal{E}_{r}}.
\end{equation}
Now note that $\frac{\partial H(x,u,t)}{\partial t}=h(x,w+u)$, and
hence $\frac{\partial H(x,0,0)}{\partial t}=h(x,w)$. Since $\{j^{k}\beta_{1}(\Xi(w,F'+ah)),\dots,j^k{k}\beta_{q}(\Xi(w,F'+ah))\}$
is a basis of $J_{n}^{k}(\Gamma)$, one has
\begin{equation}
\frac{\partial H(x,0,0)}{\partial t}=\sum_{i=1}^{q}c_{i}\beta_{i}(\Xi(w,F'+ah))+g,\label{eqder}
\end{equation}
where $g\in \mathcal{M}_{n}^{k+1}(\Gamma)$. Assume that $\vert c_{j}\vert\geq c$
for some $j$. Then, define $z=\frac{c^{2}}{c_{1}^{2}+\dots+c_{q}^{2}}<1$.
Note that $(zc_{1})^{2}+\dots+(zc_{q})^{2}=c^{2}$, i.e., $(zc_{1},\dots,zc_{q})\in S$.
But by Equation (\ref{eqskdet}) it follows that
\[
\left\Vert \sum_{i=1}^{q}c_{i}j^{k}(\beta_{i}(\Xi(w,F'+ah)))\right\Vert =\frac{1}{z}\left\Vert \sum_{i=1}^{q}zc_{i}j^{k}(\beta_{i}(\Xi(w,F'+ah)))\right\Vert \geq\frac{\delta}{2},
\]
which contradicts Equation (\ref{eqdelta2}). Hence, $\vert c_{i}\vert<c$
for all $i$. Therefore, by Equation (\ref{eqder}),
one has (since $\alpha_{i}(\Xi(w,F'+ah))=\alpha_{i}(H)$)
\[
\frac{\partial H(x,0,0)}{\partial t}=\sum_{i=1}^{n}\frac{\partial f}{\partial x_{i}}\xi_{i}+\sum_{i=1}^{r}b_{i}\alpha_{i}(H)+b_{0}+g,
\]
for some $b_{0},\dots,b_{r}\in\mathbb{R}$ and some germ $\xi=(\xi_{1},\dots,\xi_{n})\in\vec{\mathcal{E}}_n(\Gamma)$, where $\xi_{i}$ is odd in $x_{i}$ and even in the other variables.
Moreover, they also satisfy $\vert b_{0}\vert,\dots,\vert b_{r}\vert<c$
and $\vert\xi_{i}(0)\vert<c$. Since $\mathcal{M}_n^k(\Gamma)\subset \tgsp_{\mathcal{O}_f}$, we have
\[
g=\sum_{i=1}^{n}\frac{\partial f}{\partial x_{i}}\eta_{i},
\]
where $\eta\in\vec{\mathcal{E}}_n(\Gamma)$. Therefore,
\[
\frac{\partial H(x,0,0)}{\partial t}=\sum_{i=1}^{n}\frac{\partial f}{\partial x_{i}}\tilde{\xi}_{i}+\sum_{i=1}^{r}b_{i}\alpha_{i}(H)+b_{0},
\]
where $\tilde{\xi}_{i}=\xi_{i}+g_{i}$. Now, define
\begin{equation}\label{eq111}
\mu=\frac{\partial H}{\partial t}-\sum_{i=1}^{n}\tilde{\xi}_{i}\frac{\partial H}{\partial x_{i}}-\sum_{i=1}^{r}b_{i}\frac{\partial H}{\partial u_{i}}-b_{0}.
\end{equation}
Clearly $\mu(x,0,0)=0$ and hence $\mu\in \mathcal{M}_{r+1}\cdot \mathcal{U}_{n,r+1}(\Gamma)$.
Therefore, by Equation (\ref{eq59}), it follows that
\begin{equation}\label{eq222}
\mu=\sum_{i=1}^{n}\xi'_{i}\frac{\partial H}{\partial x_{i}}+\sum_{i=1}^{r}\chi_{i}\frac{\partial H}{\partial u_{i}}+\chi_{r+1}\frac{\partial H}{\partial t}+\chi_{0}.
\end{equation}
From Equations (\ref{eq111}) and (\ref{eq222}) one arrives
at
\[
\frac{\partial H}{\partial t}=\sum_{i=1}^{n}\rho'_{i}\frac{\partial H}{\partial x_{i}}+\sum_{i=1}^{r}\chi_{i}'\frac{\partial H}{\partial u_{i}}+\chi_{0}',
\]
where $\rho'_{i}=(\tilde{\xi}_{i}+\xi'_{i})/(1-\chi_{r+1})$ and $\chi_{i}'=(b_{i}+\chi_{i})/(1-\chi_{r+1})$.
Moreover, it follows that $\vert\gamma_{i}'(0)\vert<c$, $\vert\chi_{i}'(0)\vert<c$.
By Lemma \ref{difeq2}, it follows that there are some germs
$\phi$, $\psi$ and $\lambda$
such that
\begin{equation}
\begin{cases}
\frac{\partial \phi (x,0,t)}{\partial t}=-\rho'(\phi(x,0,t),0,t),\quad \phi(x,u,0)=x, \\
\frac{\partial\psi(0,t)}{\partial t}=-\chi'(\psi(0,t),t),\quad\quad\quad\:\psi(u,0)=u,\\
\frac{\partial \lambda(0,t)}{\partial t}=-\chi'_0(0,t),\quad\quad\quad\quad\quad \,\lambda(u,0)=0,
\end{cases}
\end{equation}
Satisfying
\begin{equation}
H(\phi(x,u,t),\psi(x,u,t),t)+\lambda(u,t)=H(x,u,0).
\end{equation}
I.e., $f'+ah$ at $(0,w)$ is equivalent to $f'+(a+t)h$ at $(0,w+\psi_{t}(0))$.
 \end{proof}

\section{Classification theorem for $\mathbb{Z}_{2}$-invariant germs}

In this section we prove a classification theorem for $\mathbb{Z}_2$-invariant unfoldings, similar
to Thom's classification theorem for the usual catastrophes. We have
already shown that $\Gamma$-transversal unfoldings \textit{of the same germ}
and with the same number of parameters are $\Gamma$-equivalent. The classification
theorem, in turn, states that $\Gamma$-transversal unfoldings are equivalent
to certain polynomial unfoldings, which can be more easily analyzed.
One of the most important results to prove the classification theorem
is the following version of Morse lemma, stating the existence of a symmetry-preserving
diffeomorphism that reduces the $\Gamma$-invariant functions near a non-degenarate critical point to a quadratic polynomial.

\begin{theorem}[Morse Lemma for $\Gamma$-invariant unfoldings]\label{morseinv}
Let $f\in \mathcal{E}_n(\Gamma)$ have a nondegenerate critical point at the origin, with $f(0)=0$. Then, $f$ is equivalent to its quadratic part.
\end{theorem}

\begin{proof}
Let $f(x)=h(x)+g(x)$, where $h$ is its quadratic part and $g$ is of third order. Our strategy is to prove the existence of an homotopy between $f$ and $h$; we will prove the existence of a mapping $t\mapsto \phi_t\in L_n(\Gamma)$ with $\phi_0(x)=x$ such that
\begin{equation*}
(h+tg)(\phi_t(x))=h(x), \; t\in [0,1].
\end{equation*}
Applying the classical Morse lemma to $f+tg$, it follows that $f+tg$ is equivalent (via a possibly non-$\Gamma$-equivariant family of diffeomorphism smooth in $t$) to $\pm x_1^2 \dots \pm x_n^2$, and consequently it follows that
\begin{equation}
\mathcal{M}_n^2=\{\phi\cdot\vec{\nabla}_x(h+tg)\;|\; \phi\in\mathcal{M}_n\}.
\end{equation}
In particular, since $g$ is of third order, there exists some $\xi_t\in \mathcal{M}_n$ satisfying
\begin{equation}
    \xi_t(x)\cdot\Vec{\nabla}\left(h(x)+tg(x)\right)=-g(x),
\end{equation}
Now, define 
\begin{equation}
    \eta_t(x)\doteq\int_\Gamma\gamma\xi_t(\gamma^{-1}x)\mathrm{d}\gamma,
\end{equation}
where the above integral is with respect to the Haar measure of $\Gamma$. Then, it follows that $\eta_t$ is $\Gamma$-equivariant and $\eta_t(0)=0$. Hence, taking $\phi_t$ as the solution of
\begin{equation}
 \begin{cases}
\partial_t\phi_t(x)=\eta_t(\phi_t(x)),\\
\phi_0(x)=x,
 \end{cases}
\end{equation}
the lemma follows.

\end{proof}
\begin{prop}\label{propeq} Let $F\in\mathcal{U}_{n,r}(\Gamma)$ be a $\Gamma$-transversal unfolding of a germ $f\in \mathcal{M}_{n}(\Gamma)$, and $\phi\in L_{n}(\Gamma)$.
Then, the unfolding $G(x,u)=F(\phi(x),u)$ is a $\Gamma$-transversal unfolding
of the germ $g=f\circ\phi$. \end{prop}

\begin{proof}
Let $\phi^*:\mathcal{E}_n(\Gamma)\rightarrow\mathcal{E}_n(\Gamma)$ be given by

\begin{equation}
\phi^*h=h\circ \phi.
\end{equation}
Note that, since $\phi$ is a diffeomorphism, $\phi^*$ is an isomorphism; in particular, $\phi^*(\mathcal{E}_n(\Gamma))=\mathcal{E}_n(\Gamma)$. Hence, one has

\[
\mathcal{E}_n(\Gamma)=\phi^*\left(\tgsp_{\mathcal{O}_f}+\langle1,\alpha_{1}(F),\dots,\alpha_{r}(F)\rangle_{\mathbb{R}}\right)=\tgsp_{\mathcal{O}_{f\circ \phi}}+\langle1,\alpha_{1}(F\circ\phi),\dots,\alpha_{r}(F\circ\phi)\rangle_{\mathbb{R}},
\]
i.e., $G(x,u)=F(\phi(x),u)$ is $\Gamma$-transversal.
\end{proof}

\begin{remark}
From now on, we consider $\Gamma=\mathbb{Z}_2$, where  $\mathbb{Z}_2$ acts on $\mathbb{R}^n$ in the following way:
\begin{equation*}
\gamma .x = (\gamma_1 x_1,\dots,\gamma_nx_n), \quad \gamma_i=\pm 1.
\end{equation*}
We note that the germs which are invariant under the action of $\mathbb{Z}_2$ are the ones possessing parity symmetry.
\end{remark}
\begin{lemma}[Splitting Lemma for even functions] Let $f\in \mathcal{M}_{n}^{2}(\mathbb{Z}_2)$,
with Hessian of corank $r$ at $0$. Then, there exists some $\phi\in L_{n}(\mathbb{Z}_2)$
such that

\[
f\circ\phi(x)=\hat{f}(x_{1},\dots,x_{r})\pm x_{r+1}^{2}\pm\dots\pm x_{n}^{2},
\]
where $\hat{f}(x_{1},\dots,x_{r})=f(x_{1},\dots,x_{r},0)\in \mathcal{M}_{r}^{4}(\mathbb{Z}_2)$.
\end{lemma}

\begin{proof} Since $f\in \mathcal{M}_{n}(\mathbb{Z}_2)$ and its Hessian at $0$ has
corank $r$, then without loss of generality, we can assume that in
a suitable neighborhood $V$ of $0\in\mathbb{R}^{r}$, for every $x\in V$
the Hessian of $f$ at $(x,0)$ is a matrix given by
\[
Hf(x,0)=\begin{pmatrix}\begin{matrix}\\
 & \mbox{\normalfont\Large\bfseries\ensuremath{H_{x,x}f(x,0)}}
\end{matrix} & \rvline & \mbox{\normalfont\Large\bfseries\ensuremath{0}}\\
\hline \mbox{\normalfont\Large\bfseries\ensuremath{0}} & \rvline & \begin{matrix}\pm a_{1}(x) & \cdots & 0\\
\vdots & \ddots & \vdots\\
0 & \cdots & \pm a_{n-r}(x)
\end{matrix}
\end{pmatrix},
\]
where $H_{x,x}f(0,0)$ is denegerate, $a_{i}$ is even with respect
to $x_{1},\dots,x_{r}$ for all $i$ and $a_{i}(x)>0$ for all $i$
and all $x\in V$. In particular, by Taylor's theorem, $f$ can be
written locally as
\[
f(x,y)=f(x,0)\pm y_{1}^{2}(a_{1}(x)+g_{1}(x,y))\pm\dots\pm y_{n-r}^{2}(a_{n-r}(x)+g_{n-r}(x,y)),\;\;x\in\mathbb{R}^{r},y\in\mathbb{R}^{n-r},
\]
where $g_{i}$ is even in all variables for all $i$ and $g_{i}(0,0)=0$
for all $i$. Therefore, defining
\[
\tilde{\phi}(x,y)=(x,y_{1}\sqrt{a_{1}(x)+g_{1}(x,y)},\dots,y_{n-r}\sqrt{a_{n-r}(x)+g_{n-r}(x,y)})\in L_{n}(\mathbb{Z}_2),
\]
one has

\[
f\circ\phi^{-1}(x,y)=f(x,0)\pm y_{1}^{2}\pm\dots\pm y_{n-r}^{2},
\]
and hence, by applying Theorem \ref{morseinv}, the result follows.

\end{proof}
Now we are in position to prove the classification theorems for the case of symmetric germs of funtions. We start with the classification theorem for the initial germs, associated with the equivalence relation that Definition \ref{defeq1} gives rise to. In order to prove it, we take advantage of the classification theorem exposed in \cite{siersma1974classification} for germs without symmetry.

\begin{theorem}\label{classtheo1} Let $f\in \mathcal{M}_{n}^{2}(\mathbb{Z}_2)$, with
$\cod(f)<\infty$ and Hessian of corank $r$ at $0$. Then, $\cod(f)>8$
or there exists some $\phi\in L_{n}(\mathbb{Z}_2)$ such that

\[
f\circ\phi(x)=\hat{f}(x_{1},\dots,x_{r})\pm x_{r+1}^{2}\pm\dots\pm x_{n}^{2},
\]
where $\hat{f}\in \mathcal{M}_{r}^{4}(\mathbb{Z}_2)$ is one of the following germs:

\begin{table}[H]
\centering %
\begin{tabular}{|c|c|c|c|}
\hline 
corank  & $\hat{f}$  & $\cod(f)$  & $\sigma(f)$ \tabularnewline
\hline 
$0$  & $0$  & $0$  & $2$ \tabularnewline
\hline 
$1$  & $\pm x^{2k},\;k\in\mathbb{N}$  & $2k-2$  & $2k$ \tabularnewline
\hline 
 & $\pm(x_{1}^{2}\pm x_{2}^{2})(x_{1}^{2}+\alpha x_{2}^{2}),\;\alpha\neq0,\pm1$  & $8$  & $4$ \tabularnewline
$2$  & $\pm(x^{4}+\beta x^{2}y^{2}+y^{4}),\;\beta\neq\pm2$  & $8$  & $4$ \tabularnewline
 & $x^{4}-y^{4}$  & $8$  & $4$ \tabularnewline
\hline 
\end{tabular}\caption{Classification of even germs of codimension $\protect\leq8$.}
\label{table1} 
\end{table}

\noindent Moreover, $\alpha,\beta$ are local invariants i.e., for any $\alpha\neq0,\pm1$
there exists some neighborhood $V$ of $\alpha$ such that the germs
$\{\pm(x_{1}^{2}\pm x_{2}^{2})(x_{1}^{2}+\alpha'x_{2}^{2})\;|\;\alpha'\in V\}$
are all non-equivalent to each other, and for any $\beta\neq\pm2$
there exists some neighborhood $V$ of $\beta$ such that the germs
$\{\pm(x_{1}^{4}+\beta'x_{1}^{2}x_{2}^{2}+x_{2}^{4})\;|\;\beta'\in V\}$
are all non-equivalent to each other. \end{theorem}

\begin{proof} First, note that if $f\in \mathcal{M}_{n}^{2}(\mathbb{Z}_2)$ then

\[
f(x)=a_{1}x_{1}^{2}+\dots+a_{n}x_{n}^{2}+h(x),
\]
where $O(h)\geq4$. From this, it is not hard to see that if the corank
$r$ of the Hessian of $f$ is $\geq3$, eg., $a_{1}=a_{2}=a_{3}=0$,
then $\tgsp_{\mathcal{O}_f}$ does not generate $x_{1},x_{2},x_{3},x_{1}^{2},x_{1}x_{2},x_{1}x_{3},x_{2}^{2},x_{2}x_{3},x_{3}^{2}$,
i.e., $\cod(f)\geq9$. Hence, $r\leq2$. If $r=2$, then by the splitting
lemma there exists some $\phi\in L_{n}(\mathbb{Z}_2)$ such that

\begin{equation}
f(x)=\hat{f}(x_{1},x_{2})\pm x_{3}^{2}\pm\dots\pm x_{n}^{2},
\end{equation}
where $\hat{f}(x_{1},x_{2})=f(x_{1},x_{2},0)\in \mathcal{M}_{r}^{4}(\mathbb{Z}_2)$. Moreover,
since $j^{3}\hat{f}=0$, we are under the hypothesis
of \cite[Theorem 4.6]{siersma1974classification}. note that $j^{k}\hat{f}(x)$ can be always written
as

\[
j^{k}\hat{f}(x)=ax^{4}+bx^{2}y^{2}+cy^{4},
\]
for some $a,b,c\in\mathbb{R}$. However, note that, if $a=0$ or $b=0$,
we would be in case 2, 3, 5 or 6 of \cite[Theorem 4.6]{siersma1974classification}.
And hence, by \cite[Theorems 4.6, 4.8]{siersma1974classification}, it would follow
that $\cod{\hat{f}}\geq9$. Therefore, we may restrict the analysis
to the case where $a,c\neq0$. Then, without loss of generality, $j^{4}\hat{f}$
can be assume to be

\begin{equation}
j^{4}\hat{f}(x)=\pm(x^{4}+bx^{2}y^{2}\pm y^{4}).
\end{equation}
\smallskip
\noindent Case 1: $j^{4}\hat{f}(x)=\pm(x^{4}+bx^{2}y^{2}+y^{4})$ and $b>2$. By a change of variables of the type
\[
\phi(x_{1},x_{2})=(x_{1},\beta x_{2}),\beta>0,
\]
it can be shown that $j^{4}\tilde{f}$ is equivalent to $\pm(x^{2}+y^{2})(x^{2}+\alpha y^{2})$
for some $\alpha>0$, $\alpha\neq1$.

\smallskip

\noindent Case 2: $j^{4}\hat{f}(x)=\pm(x^{4}+bx^{2}y^{2}+y^{4})$ and $b<2$. Again by a change of variables of the type
\[
\phi(x_{1},x_{2})=(x_{1},\beta x_{2}),\beta>0,
\]
it can be shown that $j^{4}\tilde{f}$ is equivalent to $\pm(x^{2}-y^{2})(x^{2}+\alpha y^{2})$
for some $\alpha<0$, $\alpha\neq-1$.

\smallskip

\noindent Note also that, if $b=\pm2$, we are in case 3 of \cite[Theorem 4.6]{siersma1974classification},
and hence $\cod(\tilde{f})\geq10$. Finally, if $r=1$, the result
is straightforward to check. \end{proof}

\begin{prop}\label{propuni} Let $F$ be a $\mathbb{Z}_2$-invariant $r$-unfolding of
$f\in \mathcal{M}_{n}(\mathbb{Z}_2)$, let $\phi\in L_{n}(\mathbb{Z}_2)$ and define the unfolding
$G(x,u)=F(\phi(x),u)$. Then, $F$ is universal if and only if $G$
is universal. \end{prop}

\begin{proof} Let $F$ be versal, and let $H$ be some $\mathbb{Z}_2$-invariant versal $d$-unfolding $H$ of $\phi(f)$. By Theorem \ref{eqtransversal}
and Proposition \ref{propeq}, $H(\phi^{-1}(x),v)$ is a versal $d$-unfolding
of $f$, and since $F$ is universal, one must have $r\leq d$.
Moreover, note also that by the same propostions, $G$ is a versal
$r$-unfolding of $\phi(f)$. Hence, it follows that $G$ is universal.
By a similar argument, the converse also holds. \end{proof}

We conclude with the classification theorem for even unfoldings.

\begin{theorem}[Classification theorem for even unfoldings]\label{theoclas2} Let $F$
be an universal unfolding of a germ $f\in \mathcal{M}_{n}(\mathbb{Z}_2)$
and Hessian of corank $r$ at $0$. Then, $\cod_{\mathbb{Z}_2}(f)\geq 4$ or $F$ is equivalent to an
unfolding $G$ of the form

\[
G(x,u)=g(x_{1},\dots,x_{r},u)\pm x_{r+1}^{2}\pm\dots\pm x_{n}^{2},
\]
where $g(x,u)$ is one of the following catastrophes:

\begin{table}[H]
\centering %
\begin{tabular}{|c|c|c|c|}
\hline 
corank  & $g(x_{1},\dots,x_{r},u)$  & $\cod_{\mathbb{Z}_2}(f)$  & $\sigma(f)$ \tabularnewline
\hline 
$0$  & $0$  & $0$  & $2$ \tabularnewline
\hline 
$1$  & $\pm x^{2k}+u_{1}x^{2}+\dots+u_{k-1}x^{2k-2},\;k\in\mathbb{N}$  & $k-1$  & $2k$ \tabularnewline
\hline 
 & $\pm(x_{1}^{2}\pm x_{2}^{2})(x_{1}^{2}+\alpha x_{2}^{2})+u_{1}x^{2}+u_{2}y^{2}+u_{3}x^{2}y^{2},\;\alpha\neq0,\pm1$  & $3$  & $4$ \tabularnewline
$2$  & $\pm(x^{4}+\beta x^{2}y^{2}+y^{4})+u_{1}x^{2}+u_{2}y^{2}+u_{3}x^{2}y^{2},\;\beta\neq\pm2$  & $3$  & $4$ \tabularnewline
 & $x^{4}-y^{4}+u_{1}x^{2}+u_{2}y^{2}+u_{3}x^{2}y^{2}$  & $3$  & $4$ \tabularnewline
\hline 
\end{tabular}\caption{Classification of $\mathbb{Z}_2$-invariant unfoldings.}
\label{table2} 
\end{table}

\end{theorem}

\begin{proof} First note that, by direct computations, it can be
shown that all of the possible unfoldings $G(x,u)$ of the theorem
are universal. By Theorem \ref{classtheo1}, it follows that there
exists some $\phi\in L_{e}(n)$ such that $f\circ\phi=\hat{f}(x_{1},\dots,x_{r})\pm x_{r+1}^{2}\pm x_{n}^{2}$,
where $\hat{f}$ is one of the germs in table \ref{table1}. Therefore,
by proposition \ref{propuni}, $F(\phi(f),u)$ is an universal unfolding
of $\hat{f}(x_{1},\dots,x_{r})\pm x_{r+1}^{2}\pm x_{n}^{2}$, and
hence, it will be equivalent to the universal unfolding of table
\ref{table2} associated with the germ $\hat{f}(x_{1},\dots,x_{r})$.

\end{proof}

\bigskip{}

\section{Applications to Quantum Statistical Mechanics}

In the field of statistical physics, be it classical or quantum, one
of the relevant goals is to find treatable analytical models that
are able to exhibit and mimic the behavior of phase transitions encountered
in nature. However, it follows that for most realistic physical models,
which generally possess interactions that are short-ranged, it is
extremely difficult to prove the existence of phase transitions. A
mathematical approach to overcome such difficulty is to find an approximation
of such interactions by what is called \textit{mean-field models}.
A mean-field model is a model where the forces of particle-particle
interaction are constant across distance. Although they seem physically
unreasonable, such forces, however, may be seen as a limit of short-range
interactions whose interaction range diverges, while its strength
goes to $0$ \cite{jbkac}. Its greatest advantage is that the mathematical
analysis of the thermodynamics of mean-field models can be carried
out analytically in a series of physically relevant scenarios: e.g.,
the Curie-Weiss model of ferromagnetism, the BCS model of superconductivity,
etc. Moreover, the results of catastrophe theory just exposed
above can be a powerful mathematical tool to carry out the analysis
of its phase transitions. Below, we provide an example of such application
in the context of Quantum Statistical Mechanics. We note that in the
following, we will only present results, mostly ommiting the proofs. The full
power of the catastrophe theory exposed in this paper to the study
of phase transitions in Q.S.M with its detailed analysis will be the subject
of a subsequent paper.

\bigskip

\noindent \textbf{Fermionic systems over a lattice:} We let $\mathbb{Z}^d$ be the $d$-dimensional cubic lattice, and $\mathfrak{A}$ be
the CAR algebra generated by the creation and annihilation elements
$\{a_{x,s}\;|\;x\in\mathbb{Z}^{d},\;s\in\{\uparrow,\downarrow\}\}$,
i.e., the unique (up to $*$-isomorphism) $C^{*}$-algebra that satisfies
the relations
\begin{align}
\{a_{x,s},a_{y,t}\}= & 0,\nonumber \\
\{a_{x,s},a_{y,t}^{*}\}= & \delta_{x,y}\delta_{s,t}\mathds{1},\quad x,y\in\mathbb{Z}_{d},\;s,t\in\{\uparrow,\downarrow\},
\end{align}
where $\{A,B\}\doteq AB+BA$ denotes the anti-commutator. From now on, $\Lambda_l\subset \mathbb{Z}^d$ will denote the cubic box of size $2l+1$ centered at the origin of $\mathbb{Z}^d$; i.e.,

\begin{equation}
\Lambda_l=\{(x_1,\dots,x_d)\in\mathbb{Z}^d\;|\; \vert x_i\vert \leq l, \; i=1,\dots,d\}.
\end{equation}

\bigskip

\noindent \textbf{Mean-field model:} We consider now a finite-dimensional
Hamiltonian on the box $\Lambda_l$ given by
\begin{align}
H_{l}= & \sum_{x\in\Lambda_{l}}(-\mu(n_{x,\uparrow}+n_{x,\downarrow}-h(n_{x,\uparrow}-n_{x,\downarrow})+2\lambda n_{x,\uparrow}n_{x,\downarrow})-\frac{\gamma}{\vert\Lambda_{l}\vert}\sum_{x,y\in\Lambda_{l}}a_{x,\uparrow}^{*}a_{x,\downarrow}^{*}a_{y,\downarrow}a_{y,\uparrow}\nonumber \\
 & -\frac{\delta}{\vert\Lambda_{l}\vert}\sum_{x,y\in\Lambda_{l};s,t\in\{\uparrow,\downarrow\}}a_{x,s}^{*}a_{x,s}a_{y,t}^{*}a_{y,t},
\end{align}
where $n_{x,s}=a^*_{x,s}a_{x,s}$. The second and third terms of the above Hamiltonian are what is called \textit{mean-field interactions}. Note that the interaction strength is constant across distances. Altough they seem at first non-physical, these kind of interactions are in fact physically relevant, specially because they can model with good precision the behavior of macroscopic materials, for example superconductors and ferromagnets. It is well known that the thermodynamics such systems with mean-field interactions may be analyzed by a simpler Hamiltonian, with no mean-field terms, usucally called in the literature as the approximating Hamiltonian (see \cite{bru2013non} for more details).

\begin{theorem}\label{resume}
The limiting thermodynamical pressure $\mathrm{P}$ of $H_l$, defined as
\begin{equation}
\mathrm{P}\doteq \lim_{l\rightarrow\infty}p_l,
\end{equation} 
where $p_l$ is the pressure of the system restricted to the finite box $\Lambda_l$, is equal to
\begin{equation}\label{eqpressure}
\mathrm{P}= \inf_{c_1,c_2 \in \mathbb{R}}\{\gamma c_1^2 +\delta c_2^2 -\mathrm{P}_0(c_1,c_2)\},
\end{equation}
and $\mathrm{P}_0(c_1,c_2)$ is the thermodynamical pressure of the system with Hamiltonian given by 
\begin{align}
H_{l}(c_{1},c_{2}) & =\sum_{x\in\Lambda_{l}}(-\mu(n_{x,\uparrow}+n_{x,\downarrow})-h(n_{x,\uparrow}-n_{x,\downarrow})+2\lambda n_{x,\uparrow}n_{x,\downarrow})-\gamma\overline{c_{1}}\sum_{x\in\Lambda_{l}}a_{y,\downarrow}a_{y,\uparrow}\nonumber \\
 & -\gamma c_{1}\sum_{x\in\Lambda_{l}}a_{x,\uparrow}^{*}a_{x,\downarrow}^{*}-\delta\overline{c_{2}}\sum_{x\in\Lambda_{l};s\in\{\uparrow,\downarrow\}}a_{x,s}^{*}a_{x,s}-\delta c_{2}\sum_{x\in\Lambda_{l};s\in\{\uparrow,\downarrow\}}a_{x,s}^{*}a_{x,s}.
\end{align}
Morover, the absolute value of $c_1$ where the infimum is achieved corresponds to the cooper pair density of the system, and the absolute value of $c_2$ where the infimum is achieved corresponds to the electron density of the system.  
\end{theorem}
Because of the permutation invariance of the system, the pressure associated with $H_l(c_1,c_2)$ can be analitically computed, and is given by
\begin{equation}\label{eq21}
\mathrm{P}_{0}(c_{1},c_{2})=-\beta_{-1}\ln2+(\mu+2\delta c_{2})+\beta^{-1}\ln\left[\cosh(\beta h)+e^{-\lambda\beta}\cosh\left(\beta\sqrt{(\mu+2\delta c_{2}-\lambda)^{2}+\gamma^{2}c_{1}^{2}}\right)\right].
\end{equation}
Therefore, in view of theorem \ref{resume}, one can study the thermodynamics and the different phases of the system by means of Equation (\ref{eqpressure}). And for this, the catastrophe theory exposed here can be of great use, since its equivalent relations \textit{preserve} the minima of functions, and hence, it preserves the information of the phase diagram of our system. With the help of numerical computations, one can show that for a particular choice of the parameters $\beta, h, \mu, \gamma,\lambda,\delta$, the pressure $P_0$ possesses a degenerate critical point at $(0,0)$, and it is moreover transversal with respect to the parameters $\gamma,\lambda,\delta$. Therefore, we may apply the classification theorem \ref{theoclas2}.

\begin{theorem}\label{theofin1}
For $\beta=1, h=0, \mu=2$, it follows that $\gamma c_1^2 +\delta c_2^2 -\mathrm{P}_0(c_1,c_2)$ is equivalent (in the sense of definition \ref{defeq2}) to the unfolding
\begin{equation}
x^{4}+\alpha x^{2}y^{2}+y^{4}+u_{1}x^{2}+u_{2}y^{2}+u_{3}x^{2}y^{2}, \quad \alpha>2.
\end{equation}
\end{theorem}

\noindent \textit{Idea of the proof:} From Equation (\ref{eq21}), we calculate the Taylor expansion of the pressure at its degenerate critical point with respect to the variables $(c_1,c_2)$, and we see that it is equivalent to the germ
\begin{equation}
f(x,y)=x^{4}+\alpha x^{2}y^{2}+y^{4}
\end{equation}
for some $\alpha\neq\pm0,1$. Then, we show that $\gamma c_1^2 +\delta c_2^2-\mathrm{P}_0(c_1,c_2)$ is $\Gamma$-transversal by means of Proposition \ref{proptransversal}, and hence , by the classification table \ref{table2}, the result follows.

\begin{remark} Note that because of the fact that our initial function $\gamma c_1^2 +\delta c_2^2 -\mathrm{P}_0(c_1,c_2)$ possesses a symmetry, we cannot not apply the standard results of catastrophe theory; we instead have to resort to the invariant version of catastrophe theory that was exposed here. We also point out that the symmetry of the pressure function above is linked with physical properties of our interaction; in special,  to the so-called \textit{gauge invariance} of $H_l$.  
\end{remark}

In view of Theorem \ref{theofin1}, one concludes that it is now analytically feasible to evaluate the phase diagram of our system apart from diffeomorphisms. Furthermore, not only the golbal minima also the local minima of the pressure mapping give us information about the system; is this case, about the so-called metastable states. To conclude, we shortly discuss about the stability of the phase diagrams of our model above. We recall that the perturbations allowed for the unfoldings to be said stable are perturbations that have their derivative controlled up to an arbitrary order $k$. And it follows that not every "small" interaction added in $H_l$ would in turn correspond to a small derivative of any order. Nevertheless, by means of the so-called cluster expansion of the partition function \cite{ueltschi2003cluster}, one is able to find a class of interactions whose derivatives can be adequately controlled to the application of catastrophe theory. One, in particular, corresponds to adding a kinetic interaction into the system.

\begin{theorem}
Let 
\begin{equation}
K_{l,\epsilon}=\sum_{\substack{x,y\in\Lambda_l\;|\;
\vert x-y\vert=1,\\ s\in\{\uparrow,\downarrow\}}}\epsilon a ^*_{x,s}a_{y,s}
\end{equation}
be a kinetic interaction. Then, it follows that, for $\epsilon>0$ small enough, the pressure of the system $H_l+K_{l,\epsilon}$ is equivalent to the unperturbed pressure. In particular, the phase diagram is preserved.
\end{theorem}
\noindent \textit{Idea of the proof:} For this result, one first needs to show that the the perturbation $K_{l,\epsilon}$ is well behaved in the sense that for $\epsilon$ small enough, the perturbation of $\mathrm{P}_0$ introduced by the addition of $K_{l,\epsilon}$ belong to some neighborhood of $\mathrm{P}_0$ in the $C^\infty$-topology of Definition \ref{deftop}. We can do that by techniques coming from the cluster expansions of the partition function studied in \cite{ueltschi2003cluster}. After that, the result follows from Theorem \ref{theostable}.

\section{Technical results}

We dedicate this section to the proof of more technical results that are necessary to prove the main theorems of the paper, but unimportant to the comprehension of the paper.

\bigskip

\noindent\textbf{Proof of Lemma \ref{lemmadet}}: Let $e\in\mathbb{N}$ be such that $\mathcal{M}_n^e(\Gamma)\subset T_{0,g}$. By item $i)$ of Theorem \ref{mlemma}, taking now $N=\mathcal{M}_n(\Gamma)$ and $M_0=\mathcal{M}_n^e(\Gamma)$, it follows that there exists some  $\omega\in\mathbb{N}$ such that, if $\mathcal{M}_n^e(\Gamma)\subset T_{\mathcal{O}_f}+\mathcal{M}_n^{\omega+e}(\Gamma)$, then $\mathcal{M}_n^e(\Gamma)\subset T_{\mathcal{O}_f}$. Choose $r=e+\omega+1$. If $j^rf=j^rg$ it is straightforward to check that
\begin{equation*}
T_{0,g}\subset T_{\mathcal{O}_f}+\mathcal{M}_n^{r-1}(\Gamma),
\end{equation*}
i.e., one has
\begin{equation}
\mathcal{M}_n^e(\Gamma)\subset T_{0,g}\subset T_{\mathcal{O}_f}+\mathcal{M}_n^{r-1}(\Gamma)= T_{\mathcal{O}_f}+\mathcal{M}_n^{\omega+e}(\Gamma),
\end{equation}
which, by $i)$ of Theorem \ref{mlemma}, implies
\begin{equation}
\mathcal{M}_n^r(\Gamma)\subset\mathcal{M}_n^e(\Gamma)\subset T_{\mathcal{O}_f}.
\end{equation}

\begin{lemma}[Geometric lemma for $\Gamma$-invariant unfoldings]\label{difeq2}
Let $H(x,u,t)\in \mathcal{U}_{n,r+1}(\Gamma)$ satisfy the following differential equation
\begin{equation}
\frac{\partial H(x,u,t)}{\partial t}=\xi(x,u,t)\cdot \vec{\nabla}_xH(x,u,t)+\chi(u,t)\cdot \vec{\nabla}_uH(x,u,t) +\chi_0(u,t),
\end{equation}
where:
\begin{align*}
&\xi(x,u,t)\in \vec{\mathcal{E}}_{n+r+1,n},\; \text{ with }x\mapsto \xi(x,u,t)\in \vec{\mathcal{E}}_n(\Gamma) \text{ for $(u,t)$ in some small neighborhood of $0\in \mathbb{R}^{r+1}$},\\
&\chi(u,t)\in \vec{\mathcal{E}}_{r+1,r},\;\text{ and}\\
&\chi_0(u,t)\in \mathcal{E}_{r+1}.
\end{align*}
Then, it follows that for small enough $t\in \mathbb{R}$, $H(x,u,t)$ is $\Gamma$-equivalent to $H(x,u,0)$, i.e., there exist some germs
\begin{align*}
&\phi(x,u,t)\in \vec{\mathcal{E}}_{n+r+1,n},\; \text{ with }x\mapsto \xi(x,u,t)\in \vec{\mathcal{E}}_n(\Gamma) \text{ for $(u,t)$ in some small neighborhood of $0\in \mathbb{R}^{r+1}$},\\
&\psi(u,t)\in \vec{\mathcal{E}}_{r+1,r},\;\text{ and}\\
& h(u,t)\in \mathcal{E}_{r+1}
\end{align*}
such that
\begin{equation}\label{eqequiv}
H(x,u,0)=H(\phi(x,u,t),\psi(u,t),t)+h(u,t)
\end{equation}
for small enough $t\in \mathbb{R}$.

\end{lemma}

\begin{proof}
We define $\phi(x,u,t)\in \vec{\mathcal{E}}_{n+r+1,n}$, $\psi(u,t)\in \vec{\mathcal{E}}_{r+1,r}$ and $\h(u,t)\in \mathcal{E}_{r+1}$ to be the (unique) germs sattisfiying the following system of ODEs:
\begin{align*}
&\frac{\partial \phi(x,u,t)}{\partial t}=-\xi(\phi(x,u,t),\psi(u,t),t),\\
&\frac{\partial \psi(x,u,t)}{\partial t}=-\chi(\psi(u,t),t),\;\text{ and}\\
&\frac{\partial h(u,t)}{\partial t}=-\chi_0(\psi(u,t),t),
\end{align*}
with initial conditions
\begin{equation}
\phi(x,u,0)=x,\quad \psi(u,0)=u,\quad h(u,0)=0.
\end{equation}
First, note that since $xi$ is $\Gamma$-invariant, for any $\gamma\in \Gamma$ it follows that $\gamma^{-1}\phi(\gamma x,u,t)$ is also a solution of the ODE with the same initial condition. Hence, by the uniqueness of the solution, one has
\begin{equation}
\gamma^{-1}\phi(\gamma x,u,t)=\phi(x,u,t),
\end{equation}
i.e., the mapping $x\mapsto \phi(x,u,t)$ is $\Gamma$-invariant for $(u,t)$ in some small neighborhood of $0\in \mathbb{R}^{r+1}$. Now we proceed to show the equality in Equation (\ref{eqequiv}). Differentiating the r.h.s.of Equation (\ref{eqequiv}) with respect to $t$, we obtain
\begin{align}\label{eq32}
\frac{\partial }{\partial t}(H(\phi(x,u,t),\psi(u,t),t)+h(u,t))=& -\xi(\phi(x,u,t),\psi(u,t),t)\cdot \vec{\nabla}_xH(\phi(x,u,t),\psi(u,t),t)\nonumber\\
&-\chi(\psi(u,t),t)\cdot \vec{\nabla}_uH(\phi(x,u,t),\psi(u,t),t) \nonumber\\
&-\chi_0(\psi(u,t),t))+\frac{\partial H }{\partial t}(\phi(x,u,t),\psi(u,t),t).
\end{align}
It follows that the r.h.s. of Equation (\ref{eq32}) is $0$, and hence, by the initial conditions, one clearly has, for small enough $t\in \mathbb{R}$,
\begin{equation}
H(\phi(x,u,t),\psi(u,t),t)+h(u,t)=H(\phi(x,u,0),\psi(u,0),0)+h(u,0)=H(x,u,0).
\end{equation}

\end{proof}

\begin{lemma}\label{lemma1} Let $f\in\mathcal{U}_{e}(n,r)$ be such
that, for some $k\leq r$, $f$ satisfies

\[
f(x,0,\dots,0,u_{k+1},\dots,u_{r})=0,
\]
for all $x\in\mathbb{R}^{n}$ in some neighborhood of $0$ and all
$(u_{k+1},\dots,u_{r})\in\mathbb{R}^{r-k}$ in some neighborhood of
$0$. Then, there exists germs $g_{1},\dots,g_{k}\in\mathcal{U}_{e}(n,r)$
such that

\[
f(x,u)=\sum_{i=1}^{k}u_{i}g_{i}(x,u).
\]
\end{lemma}

\begin{proof} Note that

\[
f(x,u)=\int_{0}^{1}\frac{d}{dt}(f(x,tu_{1},\dots,tu_{k},u_{k+1},\dots,u_{r}))dt.
\]
Computing the derivative on the right side of the equality, one has

\begin{align*}
f(x,u)= & \int_{0}^{1}\sum_{i=1}^{k}\left.\frac{\partial f}{\partial u_{i}}\right|_{(x,tu_{1},\dots,tu_{k},u_{k+1},\dots,u_{r})}\cdot u_{i}dt\\
= & \sum_{i=1}^{k}\int_{0}^{1}\left.\frac{\partial f}{\partial u_{i}}\right|_{(x,tu_{1},\dots,tu_{k},u_{k+1},\dots,u_{r})}dt\cdot u_{i}.
\end{align*}
For each $i=1,\dots,k$, define

\[
g_{i}(x,u)\doteq\int_{0}^{1}\left.\frac{\partial f}{\partial u_{i}}\right|_{(x,tu_{1},\dots,tu_{k},u_{k+1},\dots,u_{r})}dt.
\]
Since $f\in\mathcal{U}_{e}(n,r)$, it follows that $g_{i}\in\mathcal{U}_{e}(n,r)$
for all $i=1,\dots,k$, and hence, the lemma follows. \end{proof}

\subsection{Adequately ordered systems of DA algebras}

In order to prove the main theorem on the equivalence of transveral unfoldings, we need a technical result of DA algebras. The reader may find it in \cite[Corollary 6.16]{damon1984unfolding}, in a more general scenario.
 
\begin{defn}\label{defadeq}
Let $R$ be a ring and a DA-algebra, $I$ a Jacobson ideal in $R$. Then, $(R,I)$ is said to be adequate if, for any homomorphism $\Psi:N\rightarrow M$ where $M$ is finitely generated such that
\begin{equation}
\Psi(N)+I\cdot M=M,
\end{equation}
one also has
\begin{equation}
\Psi(N)=M\quad \text{and}\quad\Psi(I\cdot N)=I\cdot M.
\end{equation}
\end{defn}
\begin{remark}
Recall that a Jacobson ideal in a ring $R$ is an ideal $I$ such that $\mathbf{1}+a$ is invertible for any $a\in I$.
\end{remark}

\begin{theorem}[Preparation theorem for $\Gamma$-invariant unfoldings]\label{preptheorem}
The pair $(\mathcal{U}_{n,r}(\Gamma),\mathcal{M}_r)$ is adequate, in the sense of Definition \ref{defadeq}.
\end{theorem}

\begin{defn}
Let $F\in\mathcal{U}_{n,r}(\Gamma)$. We define
\begin{equation}
\tilde{T}_F\doteq\{\Phi\cdot\vec{\nabla}_xF\;|\; \Phi\in \vec{\mathcal{E}}_{n+r+1,n}\;\text{ and }x\mapsto \Phi(x,u,t)\in \vec{\mathcal{E}}_n(\Gamma)\}.
\end{equation}
\end{defn}

\begin{corollary}\label{coroalg}
Let $F\in \mathcal{U}_{n,r}(\Gamma)$ unfold $f\in \mathcal{M}_n(\Gamma)$. The following conditions are equivalent
\begin{enumerate}[label=\roman*.)]

\item $F$ is $\Gamma$-transversal, i.e., $\mathcal{E}_n(\Gamma) = \tgsp_{\mathcal{O}_f}+ \langle1,\alpha_1(F),\dots,\alpha_r(F)\rangle_\mathbb{R}$,

\item $\mathcal{U}_{n,r}(\Gamma)=\tilde{T}_F+\left\langle1,\frac{\partial F}{\partial u_1},\dots,\frac{\partial F}{\partial u_r}\right\rangle_{\mathcal{E}_{r}}	$.
\end{enumerate}
\end{corollary}

\begin{proof}
$ii.)\implies i.)$. This is obvious, by taking $u=0$. 

\smallskip

\noindent $i.)\implies ii.)$. We first make the following initial observation. For any $\Gamma$-invariant unfolding $G\in \mathcal{U}_{n,r}(\Gamma)$ of a germ $g\in\mathcal{E}_n(\Gamma)$, by Taylor's theorem we can always write it as
\begin{equation}
G(x,u)=g(x)+\sum_{i=1}^ru_iG_i(x,u),
\end{equation}
for some $G_i(x,u)\in \mathcal{U}_{n,r}(\Gamma)$, $i=1,\dots,r$, i.e., 
\begin{equation}\label{eqtest1}
\mathcal{U}_{n,r}(\Gamma)\cong \mathcal{E}_n(\Gamma)\oplus \mathcal{M}_r\cdot \mathcal{U}_{n,r}(\Gamma).
\end{equation}
Similarly, we note that, also by Taylor expansion, one has for any element of $\tilde{T}_F$,
\begin{align}\label{eqtest2}
&\Phi(x,u)\cdot \vec{\nabla}_xF(x,u)=\phi(x)\cdot \vec{\nabla}_xf(x)+\sum_{i=1}^ru_iG_i(x,u),\; \text{ where }\;\phi(x)\cdot \vec{\nabla}_xf(x)\in \tgsp_{\mathcal{O}_f},\nonumber\\
&\sum_{i=1}^ru_iG_i(x,u)\in \mathcal{M}_r\cdot \mathcal{U}_{n,r}(\Gamma),
\end{align}
and for any element of $\left\langle1,\frac{\partial F}{\partial u_1},\dots,\frac{\partial F}{\partial u_r}\right\rangle_{\mathcal{E}_{r}}$,
\begin{align}\label{eqtest3}
&\chi_0(u)+\sum_{i=1}^r\frac{\partial F(x,u)}{\partial u_i}\chi_i(u)=\chi_0(0)+\sum_{i=1}^r\frac{\partial F(x,0)}{\partial u_i}\chi_i(0)+\sum_{i=1}^ru_iG_i(x,u),\;\text{ where}\nonumber\\
&\chi_0(0)+\sum_{i=1}^r\frac{\partial F(x,0)}{\partial u_i}\chi_i(0)\in \langle1,\alpha_1(F),\dots,\alpha_r(F)\rangle_\mathbb{R}\;\text{ and }\; \sum_{i=1}^ru_iG_i(x,u)\in \mathcal{M}_r\cdot \mathcal{U}_{n,r}(\Gamma).
\end{align}
Now, let $N,\,M$ be given by
\begin{equation}
N=\{\Phi\in \vec{\mathcal{E}}_{n+r,n}\;|\;x\mapsto \Phi(x,u)\in \vec{\mathcal{E}}_n(\Gamma)\}\oplus \vec{\mathcal{E}}_{r,r+1},\quad M=\mathcal{U}_{n,r}(\Gamma),
\end{equation}
and define $\Psi:N\rightarrow M$ by
\begin{equation*}
\Psi(\Phi,(\chi_0,\dots,\chi_r))(x,u)=\Phi(x,u)\cdot F(x,u)+\chi_0(u)+\sum_{i=1}^r\chi_{i}(u)\frac{\partial F(x,u)}{\partial u_i}.
\end{equation*}
Clearly, $\Phi(N)= \tilde{T}_F+\left\langle1,\frac{\partial F}{\partial u_1},\dots,\frac{\partial F}{\partial u_r}\right\rangle_{\mathcal{E}_{r}}$. Hence, condition $ii.)$ is equivalent to stating that $\Psi$ is surjective. Since the pair $(\mathcal{U}_{n,r}(\Gamma),\mathcal{M}_r)$ is adequate, by Theorem \ref{preptheorem} it follows that in order to prove the surjectivity of $\Psi$, it is enough to prove
\begin{equation}\label{eqtest4}
\Psi(N)+\mathcal{M}_r\cdot \mathcal{U}_{n,r}(\Gamma)=\mathcal{U}_{n,r}(\Gamma).
\end{equation}
But from Equations (\ref{eqtest1}), (\ref{eqtest2}) and (\ref{eqtest3}), one has that the equality (\ref{eqtest4}) holds if and only if
\begin{equation}
\tgsp_{\mathcal{O}_f}+ \langle1,\alpha_1(F),\dots,\alpha_r(F)\rangle_\mathbb{R}=\mathcal{E}_n(\Gamma),
\end{equation}
which is exactly contition $i.)$.
\end{proof}

\noindent \textit{Acknowledgements}: The author would like to thank Prof. Javier Bobadilla, Prof. Maria Aparecida Ruas, and Prof. Miriam Garcia Manoel for the helpful discussions.

\newpage

 \bibliographystyle{plain}
\bibliography{lib}

\end{document}